\documentclass[12pt,reqno]{amsart}
\usepackage{xcolor}
\usepackage{amssymb,mathrsfs,bbold,comment,enumerate,tikz,caption}
\usetikzlibrary{shapes,backgrounds}
\usepackage[colorlinks=true,urlcolor=blue,citecolor=red,linkcolor=red]{hyperref}
\numberwithin{equation}{section}
\usepackage[all]{xy}

\usepackage{pifont,diagbox}
\usepackage{amsaddr}
\usepackage{apptools}
\usepackage[onehalfspacing]{setspace}
\AtAppendix{\counterwithin{theorem}{section}}
\AtAppendix{\counterwithin{equation}{section}}

\makeatletter
\renewcommand{\email}[2][]{
\ifx\emails\@empty\relax\else{\g@addto@macro\emails{,\space}}\fi
\@ifnotempty{#1}{\g@addto@macro\emails{\textrm{(#1)}\space}}
\g@addto@macro\emails{#2}
}
\makeatother

\theoremstyle{plain}
\newtheorem{theorem}{Theorem}[section]

\newtheorem*{theorem*}{Theorem}
\newtheorem{proposition}[theorem]{Proposition}
\newtheorem{corollary}[theorem]{Corollary}
\newtheorem{lemma}[theorem]{Lemma}

\theoremstyle{definition}
\newtheorem{assumption}[theorem]{Assumption}
\newtheorem{remark}[theorem]{Remark}

\newtheorem{example}[theorem]{Example}

\DeclareMathOperator{\rvar}{\overline{\rm VaR}}
\DeclareMathOperator{\VaR}{\rm VaR}
\DeclareMathOperator{\var}{\rm VaR}
\DeclareMathOperator{\es}{\rm ES}

\newcommand{\R}{\mathbb{R}} 
\newcommand{\Q}{\mathbb{Q}} 

\newcommand{\E}{\mathbb{E}}

\newcommand{\ind}{\mathbf{1}}
\renewcommand{\P}{\mathbb{P}}
\newcommand{\Norm}{\|\cdot\|}
\newcommand{\N}{\mathbb{N}}

\newcommand{\ba}{\mathbf{ba}}

\newcommand{\ca}{\mathbf{ca}}

\newcommand{\CF}{\mathcal F}
\newcommand{\CD}{\mathcal D}

\newcommand{\CX}{\mathcal X}

\newcommand{\CY}{\mathcal Y}
\newcommand{\CI}{\mathcal I}

\newcommand{\ph}{\varphi}

\newcommand{\peq}{\preceq}

\newcommand{\Entr}{\textnormal{Entr}}
\newcommand{\LC}{\mathcal{LC}}
\newcommand{\core}{\mathcal{C}}
\newcommand{\acore}{\mathcal{A}}
\newcommand{\LA}{\mathcal{LA}}
\newcommand{\up}{\mathrm{up}}
\newcommand{\lo}{\mathrm{lo}}

\usepackage[skip=0pt, indent=15pt]{parskip}

\title{Eliciting reference measures of law-invariant functionals}

\author[F.-B.~Liebrich]{\small Felix-Benedikt Liebrich}
\address{\small Amsterdam School of Economics,\\
University of Amsterdam, Netherlands}
\email{f.b.liebrich@uva.nl}

\author[R.~Wang]{\small Ruodu Wang}
\address{\small Department of Statistics and Actuarial Science,\\
University of Waterloo, Canada}
\email{wang@uwaterloo.ca}
\date{February 9, 2026}

\begin{document}

\begin{abstract}
Law-invariant functionals are central to risk management and assign identical values to random prospects sharing the same distribution under an atomless reference probability measure. This measure is typically assumed fixed.
Here, we adopt the reverse perspective: given only observed functional values, we aim to either recover the reference measure or identify a candidate measure to test for law invariance when that property is not {\em a priori} satisfied. 
Our approach is based on a key observation about law-invariant functionals defined on law-invariant domains. These functionals define lower (upper) supporting sets in dual spaces of signed measures, and the suprema (infima) of these supporting sets---if they exist---are scalar multiples  of the reference measure. In specific cases, this observation can be formulated as a sandwich theorem.
  We illustrate the methodology through a detailed analysis of prominent examples: the entropic risk measure, Expected Shortfall, and Value-at-Risk. For the latter, our elicitation procedure initially fails due to the triviality of supporting set extrema. We therefore develop a suitable modification.

\smallskip

\noindent\textsc{Keywords:} Law-invariant functionals \and reference measure \and distortion riskmetrics \and Value-at-Risk

\noindent\textsc{JEL Classification:} C60 $\cdotp$ C71 $\cdotp$ D81

\noindent\textsc{Mathematics Subject Classification (MSC2020):} 46N10 $\cdotp$ 91G70
\thanks{RW is supported by the Natural Sciences and Engineering Research Council of Canada (CRC-2022-00141, RGPIN-2024-03728).}
\end{abstract}

\maketitle


\setstretch{1.15}

\section{Introduction}
\label{intro}
Law-invariant functionals are omnipresent in economics, finance, and risk management, in particular because they allow for standard statistical analysis. 
Mathematically speaking, a functional $\ph$ on a set $\mathcal D$ of random variables over a probability space $(\Omega,\Sigma,\P)$ is called law invariant with respect to $\P$ if it assigns the same value to random variables in $\mathcal D$ with the same distribution under $\P$. 
The probability measure 
$\P$ is then called a reference measure for $\ph$.
To ensure a seamless mathematical study of law-invariant functionals,  $\P$ is typically required to be atomless.

Law-invariant risk measures have been widely studied in the mathematical finance literature. Earlier contributions on their  theoretical foundation include Kusuoka \cite{K01}, Frittelli and Rosazza Gianin \cite{Frittelli}, Weber~\cite{W06}, Filipovi\'c and Svindland~\cite{FilSvi}, and Bellini et al.\ \cite{Bellini}, among others. 
For insights into their statistical properties, see, e.g., Cont et al.\ \cite{Cont}, Shapiro \cite{Shapiro}, and Kr\"atschmer et al.~\cite{KSZ}. A comprehensive review is provided by He et al.\ \cite{Review}.
Law-invariant functionals are widely used in contexts  beyond risk measures, such as  economic decision principles (e.g., Yaari \cite{Yaari}), insurance premia (e.g., Wang et al.~\cite{Insurance}), and deviation measurement (e.g., Rockafellar et al.~\cite{RUZ06}).

There are two common but different ways of formulating law-invariant functionals. 
The first formulation, as mentioned above, treats  them as mappings $\ph$ from $\mathcal D$ to $\R$ (or the extended real line), as in, e.g., Frittelli and Rosazza Gianin~\cite{Frittelli} and Kusuoka~\cite{K01}. 
The second formulation  treats them as mappings $\phi$  from a set of distributions  to $\R$, as in, e.g., Kr\"atschmer et al.~\cite{KSZ} and Weber~\cite{W06}.
These two approaches are often argued as being equivalent, and connected via the relation $\ph(X)=\phi (\mathbb P\circ X^{-1})$ for $X\in \mathcal D$.
It is clear that the above equivalence relies on the specification of  a probability measure $\P$ that is assumed  fixed and known.  
This reflects standard practice in risk measurement for 
regulatory capital calculations, portfolio optimisation, performance analysis, and capital allocation. 

In contrast to the procedure of first  specifying the probability measure and then computing the  value of the risk measure, this paper takes the {\em reverse} perspective: we assume that the values of the functional $\ph$ are observable, but the underlying reference measure is unknown. 
Our setting contains another layer of agnosticism: we may not know if the functional is law invariant to begin with. Liebrich \cite{Liebrich} presents general conditions under which reference measures are unique, provided they exist. 
Our aim is to recover   the (often unique) reference probability measure from a potentially black-box risk measurement procedure (henceforth, ``the functional''), if such a measure exists. 
The procedure of identifying the probability measure underlying a functional is sometimes called elicitation (e.g., Kadane and Winkler \cite{KW88}), but 
it should not be confused with the literature on the {\em elicitability} of risk measures, such as Ziegel~\cite{Z16}, Kou and Peng~\cite{KP16}, Fissler and Ziegel~\cite{FZ16}, and Embrechts et al.~\cite{EMWW21}, where elicitability means risk being a minimiser of an expected loss.

In sum, the goal of this paper is twofold: 
\begin{itemize}
\item[(a)]
to identify a procedure as general and unifying as possible that allows one to elicit the reference measure of law-invariant functionals; and 
\item[(b)] in cases where law invariance is not assumed, to produce a candidate measure for which one can test law invariance. 
\end{itemize}

To illustrate the setting, consider a regulatory authority evaluating a large set of risk estimates provided by a financial institution for payoffs modelled by random variables in a set $\mathcal D$. 
While the regulator may prescribe the risk measure---such as Expected Shortfall under the current Basel Accords---the institution's internal probability model may be unknown or not truthfully disclosed. 
This aligns with the framework of Fadina et al.~\cite{Fadina}, which addresses risk measurement under unfixed probabilities. 
Without information about the internal model, a realistic assessment of the institution's risk management practices may not be possible.
While this illustration motivates our investigation, we emphasise that our results are purely theoretical and do not, e.g., provide an algorithmic approach to addressing the aforementioned issue.

The elicitation of probability measures from observable decisions is a classical topic in decision theory, dating back to Ramsey~\cite{Ramsey}, de Finetti~\cite{Finetti}, Savage~\cite{Savage}, and Anscombe and Aumann~\cite{AA}; see also Machina and Schmeidler~\cite{MJ} for an approach beyond expected utility. 
Also, ``law invariance'' goes under the name of ``probabilistic sophistication'' in that field, often with some additional minor properties (like monotonicity). 
Our methodology diverges from these established theories by proposing a broadly applicable, model-agnostic procedure that does not rely on any particular decision-making framework.
This involves taking the numerical representations of preferences as given and focusing on the observable functional values.

To convey the core idea of our approach, we consider the prominent case of the Expected Shortfall (ES) risk measure as an illustration.
Let us recall that the ES of a bounded random variable $X$ under the probability measure $\P$ and for a given level $\alpha\in[0,1)$ is
\begin{equation*}\label{eq:ESintro}\es_\alpha^\P(X)=\frac 1 {1-\alpha}\int_\alpha^1\var_s^\P(X){\rm d}s,\end{equation*}
where for $s\in(0,1)$
\begin{equation}\label{eq:VaRintro}
\var_s^\P(X)=\inf\{x\in\R\colon \P[X\le x]\ge s\}.\end{equation}
Let $A$ be an event and denote its indicator by $\ind_A$.
If the $\P$-probability of $A$ is small enough, i.e., less than $1-\alpha$, then 
$$\es^\P_\alpha(\ind_A)=\frac{\P[A]}{1-\alpha}.$$
Hence, if the reference measure $\P$ is atomless, it can be found by splitting each event $B$ into sufficiently many pairwise disjoint small pieces $A_1,\dots,A_n$ and computing 
$$\P[B]=(1-\alpha)\sum_{i=1}^n\es^\P_\alpha(\ind_{A_i}).$$
Of course, it is {\em a priori} unclear what ``sufficiently small'' means in this case, as the answer would require knowledge of both $\P$ and $\alpha$. 
Without this information, the splitting operation should be understood as a limiting procedure. 

Now we take a dual perspective and look at the set $\mathcal L$ of all {\em measures} $\mu$ which respect the ES constraint, i.e., which satisfy the inequality
$\mu(A)\le \es_\alpha^\P(\ind_A)$ for all events $A\in\Sigma$.
We will call such a set a {\em supporting set}. 
Thus, if $\P[A]$ is sufficiently small, we also have 
$$\mu(A)\le \frac{\P[A]}{1-\alpha}.$$
This suggests that the measure $\frac 1{1-\alpha}\P$ could be the {\em least upper bound} (the supremum) of the set $\mathcal L$ in the space of measures. 
We shall not only confirm this conjecture, but in a nutshell show that suprema of such supporting sets can in many cases be used to elicit $\P$. 
Moreover, if it is unknown whether a functional is law invariant, but an atomless measure can be computed as a supremum as above, then the probability measure obtained by normalisation is the only candidate for which law invariance needs to be checked.

In summary, we shall establish a clear dual link between large classes of functionals and their reference probabilities, including the observation that the latter are ``dual'' and not ``primal objects''.

The paper is organised as follows. 
Sect.~\ref{sec:first} contains the main results  on general law-invariant functionals.
It establishes that the suprema of lower supporting sets (as in the ES example) and the infima of upper supporting sets of law-invariant functionals on law-invariant domains are directly linked to the reference measure, in the sense that they are scalar multiples of it.
We also examine the role of countable vs.\ finite additivity, i.e., computing suprema and infima in the dual space of bounded random variables, and show that our results remain stable even under finitely additive reference probabilities. 

In Sect.~\ref{sec:distortion}, we focus on distortion riskmetrics in the sense of Wang et al.~\cite{Riskmetrics}. 
These functionals allow a significant reduction in complexity, as they are fully determined by their values on the domain of indicator random variables. 
However, this domain is need not be law invariant.
In this context, a law-invariant functional corresponds to a law-invariant cooperative game, with supporting sets known as the (loose) core and (loose) anticore, following Lehrer and Teper~\cite{Lehrer}. 
We present direct analogues of the results from Sect.~\ref{sec:first}, develop a geometric interpretation as sandwich theorems \`a la Kindler~\cite{Kindler} in case of sub- and superadditive games, and highlight key caveats.

The prominent examples of entropic risk measure, Expected Shortfall, and the Value-at-Risk are presented in Sect. \ref{sec:ex}. 
There, we also illustrate the dependence of infimum/supremum of supporting sets on the domain of definition of the functional. 
The VaR case is notable because infimum and supremum of both supporting sets exist but are trivial, revealing nothing about the reference probability. Given VaR’s practical importance in risk management, a suitably modified elicitation approach is developed in Sect.~\ref{sec:VaR}. 
Some mathematical preliminaries and auxiliary results are relegated to the appendix.

\section{Notation and preliminaries} 

Throughout this paper, $\Omega$ is a nonempty set of states (or scenarios), and $\Sigma \subseteq 2^\Omega$ denotes the collection of observable events.  We assume at minimum that $\Sigma$ is an algebra, and sometimes adopt the stronger, but standard assumption that $\Sigma$ is a $\sigma$-algebra.

A signed charge is a set function $\mu \colon \Sigma \to \mathbb{R}$ that is additive, meaning
$\mu(A \cup B) = \mu(A) + \mu(B)$ for all disjoint events $A,B \in \Sigma$.
In this paper, signed charges stem from the spaces $\ba$ of bounded signed charges or $\ca$ of bounded signed measures; see Appendix~\ref{app:definitions} below for more details on these two spaces. As standard notation in functional analysis, $\ba$ stands for \textbf bounded \textbf additive functions on $\Sigma$
and $\ca$ stands for bounded \textbf countably \textbf additive functions on $\Sigma$. 

Elements of $\ba$ and $\ca$ will be denoted by Greek letters such as $\mu$ and $\nu$.
A probability charge $P$ is a signed charge that is nonnegative and satisfies $P[\Omega]=1$; we denote such charges by $P$ or $Q$.  
When $\Sigma$ is a $\sigma$-algebra, we also consider countably additive probability charges (i.e., true probability measures), which we denote by $\P$ or $\Q$ to emphasise countable additivity.
A probability charge $P$ on $\Sigma$ has convex range if, for all $A\in\Sigma$,
\begin{center}$\{P[B]\colon B\in\Sigma,~B\subseteq A\}=\big[0,P[A]\big].$\end{center}
This property is called {\em strong nonatomicity} in \cite[Definition 5.1.5]{BR}, which reflects that a probability measure has convex range if and only if it is atomless. 

If $\Sigma$ is a $\sigma$-algebra, the space of all bounded, $\Sigma$-measurable and real-valued functions is denoted by $B(\Sigma)$. 
If $\Sigma$ is only an algebra, we avoid technicalities and instead consider the space $B_s(\Sigma)$ of all simple random variables, i.e., the span of all indicator functions of events in $\Sigma$. 
Both spaces can be equipped with the supremum norm $\Norm_\infty$.
Their elements will be referred to as random variables and denoted by capital letters. 

Throughout the paper, we will impose:  
\begin{assumption}\label{ass:structure}
    Either
    \begin{enumerate}
        \item[(A)] $\Sigma$ is a $\sigma$-algebra, functionals are defined on random variables in $B(\Sigma)$, and an atomless probability measure $\P$ on $\Sigma$ will be of particular interest; or
        \item[(B)] $\Sigma$ is an algebra, functionals are defined on random variables in $B_s(\Sigma)$, and a convex-ranged probability charge $P$ will play a prominent role. 
    \end{enumerate}
\end{assumption}
Note that we always tacitly assume that $\Sigma$ is a $\sigma$-algebra when dealing with countable additivity.

Two random variables $X,Y\in B(\Sigma)$ (or $B_s(\Sigma)$) are equally distributed under $P$ (denoted $X\sim_P Y$) if, for all intervals $I\subseteq \R$, $P[X\in I]=P[Y\in I]$.
If the probability charge is a probability measure $\P$, $X \sim_\P Y$ holds if and only if the Borel probability measures $\P \circ X^{-1}$ and $\P \circ Y^{-1}$ on $\R$ agree---that is, for every Borel set $A \subseteq \R$, we have $\P[X \in A] = \P[Y \in A]$. 
 
Moreover, for an event $A\in\Sigma$ with $P[A]>0$, $P^A\colon \Sigma\to[0,1]$ denotes the conditional probability charge defined by
$$P^A[B]=\frac{P[A\cap B]}{P[A]}.$$
If $\P$ is a probability measure, a statement holds {\em $\P$-almost surely} ($\P$-a.s.)~if it holds with $\P$-probability 1. 

The integral of a simple random variable with respect to $\mu\in\ba$ will be denoted by $\E_{\mu}$. Similarly, if $\Sigma$ is a $\sigma$-algebra, we can extend $\E_{\mu}$ to the Dunford-Schwartz integral of a bounded random variable. We refer to \cite[Sections 11.1--11.2]{Ali}.

Let $\CX$ be an $\R$-vector space. 
A reflexive, transitive, antisymmetric binary relation $\peq$ on $\CX$ is a vector space order if
$X\peq Y$ implies for all $t\ge 0$ and all $Z\in\mathcal X$ that $tX+Z\peq tY+Z$. An ordered vector space $(\CX,\peq)$ is a vector lattice if all $X,Y\in\mathcal X$ have a supremum, a least upper bound with respect to $\peq$. 

Let $\CY \subseteq \CX$ be nonempty. An element $Y^\star \in \CX$ is called an {\em upper bound} of $\CY$ if $Y \peq Y^\star$ for all $Y \in \CY$. The set $\CY$ is said to be {\em upper bounded} if such an upper bound exists. The {\em supremum} of $\CY$, if it exists, is the least element among all upper bounds of $\CY$. 
Analogously, the {\em infimum} of $\CY$ is the greatest lower bound, if it exists.
For more information, we refer the reader to \cite[Chapters 8--10]{Ali}.

The most important vector lattices in this paper will be $\ba$ and $\ca$ which are endowed with the setwise order, i.e., $\mu\le \nu$ if $\mu(A)\le \nu(A)$ holds for all $A\in\Sigma$. 
In particular, every upper (lower) bounded subset of $\ba$ has a supremum (infimum); see \cite[Theorems 8.24 \& 9.11]{Ali}. 
Function spaces like $B_s(\Sigma)$ and $B(\Sigma)$ are also vector lattices when endowed with the pointwise order.

\section{Supporting sets identify the reference measure}\label{sec:first}

\subsection{Countably additive reference measures}\label{sec:countable}

Throughout this subsection, we impose Assumption~\ref{ass:structure}(A), i.e., $\Sigma$ is a $\sigma$-algebra and $\P$ is an atomless probability measure. Also, let $\mathcal D\subseteq B(\Sigma)$ be nonempty.
A functional $\ph\colon \CD\to\R$ is invariant with respect to $\P$ (or $\P$-invariant) if 
$$X,Y\in \mathcal D\text{ and }X\sim_\P Y\quad\implies\quad \ph(X)=\ph(Y).$$
Similarly, a set $\mathcal D\subseteq B(\Sigma)$ is invariant with respect to $\P$ (or $\P$-invariant) if 
$$X\in\CD,~Y\in B(\Sigma)\text{ and }X\sim_\P Y\quad\implies\quad Y\in\CD.$$
In this context, we shall refer to $\P$ as the {\em reference (probability) measure}.

Two important objects for our study are the \emph{lower supporting set} 
\begin{equation}\label{eq:L}\mathcal L:=\{\mu\in\ca\colon\E_\mu[X]\le \ph(X)\text{ for all }X\in\mathcal D\}\end{equation}
and the {\em upper supporting set}
\begin{equation}\label{eq:U}\mathcal U:=\{\mu\in\ca\colon\E_\mu[X]\ge \ph(X)\text{ for all }X\in\mathcal D\}\end{equation}
of $\ph$.
They collect all signed measures whose integrals over $\mathcal D$ respect the constraint given by $\ph$ throughout.

Note that, like $\P$, every charge in $\mathcal L$ and $\mathcal U$ is countably additive, reflecting the general preference for countable additivity in the literature. 
Nonatomicity of $\P$, however, implies that $\Sigma$ cannot have finite/countable cardinality, and the norm dual space of $B(\Sigma)$ is, up to an isometric isomorphism, actually given by $\ba$ and significantly larger than $\ca$.
This suggests that finite  additivity seems more consistent with the underlying measurable structure.
 
Recognising this discrepancy naturally leads to the {\em a priori} larger alternative supporting sets 
\begin{equation}\label{eq:Lf}
    \mathcal L^f:=\{\mu\in\ba\colon~\E_\mu[X]\le \ph(X)\text{ for all }X\in\CD\}\supseteq\mathcal L
\end{equation}
and 
\begin{equation}\label{eq:Uf}
    \mathcal U^f:=\{\mu\in\ba\colon~\E_\mu[X]\ge \ph(X)\text{ for all }X\in\CD\}\supseteq\mathcal U.
\end{equation}
The superscript $f$ in the notation emphasises that the elements in the respective set are only finitely additive.

Our first main result gives a necessary condition for $\P$ to be a reference measure for $\ph$, expressed in terms of supporting sets.
Remarkably, the condition remains the same whether one considers the countably additive supporting sets $\mathcal L$ and $\mathcal U$ or the finitely additive supporting sets $\mathcal L^f$ and $\mathcal U^f$.

\begin{theorem}\label{main1}
Let  $\mathcal D\subseteq B(\Sigma)$ and $\ph\colon \CD\to\R$ both be invariant with respect to the atomless probability measure $\P$. 
Let $\mathcal L,U,\mathcal L^f$ and $\mathcal U^f$ be the associated supporting sets defined by \eqref{eq:L}--\eqref{eq:Uf}.
\begin{enumerate}
\item[(i)] 
The supremum $\sup\mathcal L^f$ exists in $\ba$ if and only if the supremum $\sup\mathcal L$ exists in $\ca$, in which case there is a constant $a\in\R$ such that 
\[a\P=\sup\mathcal L=\sup\mathcal L^f\in \ca.\]
\item[(ii)] Likewise, the infimum $\inf\mathcal U^f$ exists in $\ba$ if and only if the infimum $\inf\mathcal U$ exists in $\ca$, in which case there is a constant $b\in\R$ such that 
\[b\P=\inf\mathcal U=\inf\mathcal U^f\in \ca.\]
\end{enumerate}
\end{theorem}

Note that Theorem~\ref{main1} does not claim that $\sup\mathcal L$ is an element of $\mathcal L$ itself; we may well find $X\in\mathcal D$ such that $\ph(X)<\E_{\sup\mathcal L}[X]$. 
Likewise, we often have $\inf\mathcal U\notin\mathcal U$, i.e., $\ph(X)>\E_{\inf\mathcal U}[X]$ holds for some $X\in\mathcal D$.
Moreover, the result fails without assuming nonatomicity of $\P$, as the subsequent example shows.

\begin{example}\label{ex1}
Consider $\Omega:=\{0,1\}$ together with the power set $\Sigma=2^\Omega$. 
Moreover, suppose $\P[\{0\}]=1-\P[\{1\}]=\frac 2 3$.
One verifies that $X,Y\in B(\Sigma)$ satisfy $X\sim_\P Y$ if and only if $X=Y$.
This shows that every nonempty set $\mathcal D\subseteq B(\Sigma)$ of such random variables is $\P$-invariant. 
Likewise, the functional $\ph\colon\mathcal D\to\R$ defined by 
\begin{center}$\ph(X)=X(1)$\end{center}
is $\P$-invariant independent of the domain of definition $\mathcal D$. 
Now suppose that $\mathcal D$ contains the random variables
$X_t:=t\ind_{\{0\}}+\ind_{\{1\}}$, $t\in\R$, and let $\mu\in\mathcal L$.
Then
\[
1=\ph(X_t)\ge \sup_{t\in\R}\E_\mu[X_t]
=\mu(\{1\})+\sup_{t\in\R}t\mu(\{0\}),
\]
implying $\mu(\{0\})=0$ and $\mu(\{1\})\le 1$.
Analogously, any $\mu\in\mathcal U$ satisfies $\mu(\{0\})=0$ and $\mu(\{1\})\ge 1$.
Hence, the probability measure $\P^\star$ placing full mass on $\{1\}$ is the supremum (infimum) of the lower (upper) supporting set and not a multiple of $\P$. 
\end{example}

Theorem~\ref{main1} is the first main result of the paper and therefore warrants a thorough discussion before we proceed with the mathematical development.
First, how could one (hypothetically) use Theorem~\ref{main1} to elicit the reference measure $\P$?  
Complete knowledge of the values that $\ph$ takes on $\mathcal{D}$ can allow us to identify all linear and (order-)\-continuous operators $\E_\mu[\cdot]$ that satisfy the constraint imposed by $\ph$.
This leads to the supporting sets $\mathcal L,\mathcal U,\mathcal L^f$, and $\mathcal U^f$, respectively. For the sake of simplicity, we focus on the sets $\mathcal L$ and $\mathcal U$ of signed measures in this illustration. 
If the supremum of $\mathcal{L}$ (or the infimum of $\mathcal{U}$) exists and is denoted by $\mu^\star$, three cases arise:
\begin{enumerate}
\item[Case 1:] $\mu^\star = 0$: No meaningful conclusion can be drawn about the reference measure.
\item[Case 2:] $\mu^\star \neq 0$, but is not a multiple of an atomless $\P$ under which $\mathcal{D}$ is invariant: then $\ph$ cannot be law invariant.
\item[Case 3:] $\mu^\star = c\P$ for some $c \in \R \setminus\{0\}$ and an atomless $\P$ under which $\mathcal{D}$ is invariant: then $\P$ uniquely identifies the only viable candidate for a reference measure consistent with $\ph$.
We distinguish two sub-cases:
\begin{enumerate}
\item[(a)] If $\ph$ is known to be law invariant, the reference measure must be $\P$.
\item[(b)] If law invariance is not assumed, it suffices to test $\ph$ for invariance under $\P$. For example, one can try to find $X, Y \in \mathcal{D}$ with the same $\P$-distribution but $\ph(X) \neq \ph(Y)$ to disprove $\P$-invariance.
\end{enumerate}
\end{enumerate}

Second, we present an illustrative example to help orient the reader.
\begin{example}\label{ex:RM}
Suppose that $\ph$ is a positively homogeneous and cash-additive risk measure on $\CD=B(\Sigma)$, i.e., 
\begin{itemize}
    \item $\ph(tX)=t\ph(X)$ for all $(X,t)\in B(\Sigma)\times(0,\infty)$, 
    \item $\ph(X+c)=\ph(X)+c$ for all $(X,c)\in B(\Sigma)\times \R$.
\end{itemize}
Its {\em convex conjugate} is given by 
\begin{equation}\label{def:conv conj}\ph^*(\mu):=\sup\{\E_\mu[X]-\ph(X)\colon X\in B(\Sigma)\},\quad \mu\in\ba.\end{equation}
If instead of the norm dual $\ba$ one prefers to use the $\langle B(\Sigma),\ca\rangle$-dual pairing, one can restrict $\ph^*$ to $\ca$. 
The lower supporting set $\mathcal L^f$ is just the effective domain of $\ph^*$, i.e., the set $\{\ph^*<\infty\}$. Analogously, $\mathcal L$ is the effective domain of the restricted convex conjugate. 
The latter is sometimes referred to as the ``scenario set'' associated with $\ph$, as it consists only of probability measures that represent relevant probabilistic scenarios for the computation of $\ph$ by dual representation.

Now suppose the values of $\varphi$ are only known on the smaller $\P$-invariant domain 
\begin{equation}\label{widetilde D}
    \widetilde{\CD}=\{X\in B(\Sigma)\colon X\ge 0~\P\text{-a.s.}\}\subsetneq\mathcal D=B(\Sigma).
\end{equation}
It is usually understood as representing true losses, in contrast to net losses with upside potential.
The strict inclusion makes the lower supporting set $\widetilde{\mathcal L}$ associated with $\ph|_{\widetilde{\mathcal D}}$ bigger than $\mathcal L$. 
Also, the signed measures $\mu\in\widetilde{\mathcal L}$ will not automatically satisfy the normalisation $\mu(\Omega)=1$ anymore.
\end{example}

The situation in the preceding example that a risk measure is only evaluated on true losses rather than all bounded net losses is common. 
While the elements of, e.g., the lower supporting set are not all interpretable as probabilistic scenarios in this case anymore, such an interpretation is desirable in decision frameworks dealing with uncertainty. 
Thus, it is only natural to ask if the assertion of Theorem~\ref{main1} is preserved if a normalisation constraint that maintains the link with probabilistic scenarios---like $\E_\mu[1]=1=\ph(1)$---is added to the definition of the supporting sets. Note that this addition is precisely what distinguishes the loose (anti)core of a game from its (anti)core in Sect.~\ref{sec:distortion} below.

\begin{corollary}\label{cor:remark}
In the situation of Theorem~\ref{main1}, fix a constant $c\in \CD$. 
Then Theorem~\ref{main1} holds verbatim 
 if $\mathcal L,\mathcal U, \mathcal L^f$ and $\mathcal U^f$ are replaced by 
\begin{align*}
&\mathcal L_c
= \{\mu \in \mathcal L \colon \E_\mu[c] = \ph(c)\},
&\mathcal U_c
= \{\mu \in \mathcal U \colon \E_\mu[c] = \ph(c)\},
\\
&\mathcal L_c^f
= \{\mu \in \mathcal L^f \colon \E_\mu[c] = \ph(c)\},
&\mathcal U_c^f
= \{\mu \in \mathcal U^f \colon \E_\mu[c] = \ph(c)\}.
\end{align*}
\end{corollary}

Third, as explained above, both notions of supporting sets---$\sigma$-additive and finite\-ly additive---are grounded in sound considerations. 
Signed measures are often easier and more convenient to work with, while charges are more natural from a functional-analytic perspective. 
Hence, an additional message of Theorem~\ref{main1} is that, provided an atomless and countably additive reference probability measure $\P$ exists, these two perspectives are equivalent and yield the same results in our approach.
This observation mirrors known results on the ``automatic Fatou property'' of law-invariant (quasi)convex functionals; see Jouini et al.~\cite{JST}, Liebrich and Munari \cite{LiebrichMunari} and Svind\-land~\cite{Svindland}. 
These reveal that the dual description of such functionals defined on bound\-ed random variables only requires the $\sigma$-additive elements of the space $\ca$ and not the full dual space $\ba$. 
In the same spirit, our result demonstrates that computing supremum or infimum in the smaller space $\ca$ or the larger space $\ba$ does not entail any differences.

Fourth, the application of Theorem~\ref{main1} outlined above treats the reference probability measure of $\ph$ as {\em a priori} unknown, but draws heavily from the  assumption that the primal domain $\CD$ is already $\P$-invariant. 
While this may seem paradoxical, it is a crucial assumption that cannot be dropped without risking that the result fails. 

\begin{example}\label{ex2}
    Let $\P,\Q$ be two atomless probability measures on $(\Omega,\Sigma)$ that are equivalent (i.e., they share the same set of null events) but such that their maximum $\P\vee\Q$ is linearly independent both of $\P$ and $\Q$. This is the case if the $\P$-density of $\Q$ is nonconstant and not bounded away from 0.  
    Define the $\Q$-invariant domain of definition 
    $$\mathcal D:=\{X\in B(\Sigma)\colon \E_\Q[X]\le 0\text{ or }X\text{ is }\Q\text{-a.s.\ constant}\}$$
    and the $\P$-invariant functional 
    $\ph\colon\mathcal D\to\R$ by
    $$\ph(X):=\max\{\E_\P[X],0\}.$$
    We claim that $\mathcal L$ coincides with the convex hull co$(\{0,\P,\Q\})$.
    The inclusion ``$\subseteq$" holds by construction.  
    Conversely, let $\mu\in\mathcal L$.
    As $1\in\mathcal D$, we obtain $\mu(\Omega)=\E_\mu[1]\le 1$. 
    Assume towards a contradiction that $\mu\notin\mathrm{co}(\{0,\P,\Q\})$. We can then find $X\in B(\Sigma)$ such that 
    $0\le s:=\max\{0,\E_\P[X],\E_\Q[X]\}<\E_\mu[X]$.
    Now, the random variable $Y:=X-s$ lies in $\mathcal D$,
    $\ph(Y)=0$, and $\E_\mu[Y]=\E_\mu[X]-s\mu(\Omega)>s\big(1-\mu(\Omega)\big)\ge 0$,
    contradicting that $\mu\in\mathcal L$. 
    Summing up, $\mathcal L=\mathrm{co}(\{0,\P,\Q\})$ and $\sup\mathcal L=\P\vee\Q$, which is not a multiple of $\P$ or $\Q$ by assumption. 
\end{example}

While the results in Sect.~\ref{sec:distortion} below are more robust against the criticism of {\em a priori} knowledge about the $\P$-invariance of $\mathcal D$, we can already offer an argument here to address potential concerns.
For example, a domain like $\widetilde{\CD}$ in equation~\eqref{widetilde D} is a natural choice and simultaneously $\Q$-invariant with respect to any probability measure $\Q$ equivalent to the true reference measure $\P$. 
Thus, the only information required about the reference measure in advance is the {\em equivalence class} to which it belongs.

\subsection{Finitely additive reference probabilities}

As explained in the preceding subsection, finitely additive charges may be more natural in a setting where the $\sigma$-algebra $\Sigma$ is rich.
This motivated the introduction of the supporting sets $\mathcal L^f$ and $\mathcal U^f$. We now go one step further by replacing countable additivity of the reference probability measure with finite additivity.

In this subsection, we impose Assumption~\ref{ass:structure}(B), i.e., $\Sigma$ is a mere algebra rather than a $\sigma$-algebra and $P$ is a convex-ranged (finitely additive)  probability charge on $\Sigma$.
Random variables $X,Y\in B_s(\Sigma)$ are said to be equidistributed under $P$ ($X\sim_P Y$) if, for all $x\in\R$, $P[X=x]=P[Y=x]$ holds. 
$P$-invariance of a subset $\CD\subseteq B_s(\Sigma)$ or a functional $\ph$ are defined in analogy to $\P$-invariance. In this case, we call $P$ the {\em reference probability}.  

Theorem~3.6 extends Theorem~3.1 to the setting of finitely additive reference probabilities, replacing the assumption of a countably additive atomless reference measure. 
All previous conclusions remain valid, making this theorem the second main result of our paper.

\begin{theorem}\label{main4}
Let  $\mathcal D\subseteq B_s(\Sigma)$ and $\ph\colon \CD\to\R$ both be invariant with respect to the convex-ranged probability charge $P$. 
Moreover, define the supporting sets $\mathcal L^f$ and $\mathcal U^f$ by \eqref{eq:Lf} and \eqref{eq:Uf}, respectively. 

\begin{enumerate}
\item[(i)]  If 
$\sup\mathcal L^f$ exists in $\ba$, then there is a constant $a\in\R$ such that 
\[\sup\mathcal L^f=a P.\]
\item[(ii)] Likewise, if $\inf\mathcal U^f$ exists in $\ba$, then there is a constant $b\in\R$ such that 
\[\inf\mathcal U^f=b P.\]
\end{enumerate}
\end{theorem}

Example~\ref{ex1} and a slight variation of Example~\ref{ex2} also show that in this case none of the assumptions can be dropped. 

\subsection{Proofs of the results in this section}

We now prove the preceding main results, starting with Theorem~\ref{main1}. 
A first step is the following  crucial lemma, which may be of some independent interest.
For the sake of convenience we work with the well-known space $L^0_\P$ of equivalence classes of all real-valued random variables up to $\P$-a.s.\ equality. 
Equivalence classes themselves will, like random variables, be denoted by capital letters, and inequalities between them are assumed to hold $\P$-a.s.\ The notion of $\P$-invariance introduced in Sect.~\ref{sec:countable} immediately transfers to this setting. 

\begin{lemma}\label{lem:constant}
Suppose $\mathcal Z\subseteq L^0_\P$ is a $\P$-invariant set of equivalence classes of random variables that is upper bounded, i.e., there exists $Y\in L^0_\P$ such that $Z\le Y$ holds for all $Z\in\mathcal Z$.
Then $\sup\mathcal Z$ exists and is constant $\P$-a.s.\
In particular, if $Z\ge 0$ $\P$-a.s.\ for all $Z\in\mathcal Z$, then $\mathcal Z$ contains only $\P$-a.s.~bounded random variables.
\end{lemma}
\begin{proof}
Denote by $\mathcal U$ the set of all $U\in L^0_\P$ with the uniform distribution over $(0,1)$ under $\P$.
The existence of $Z^\star:=\sup\mathcal Z$ is guaranteed by \cite[Theorem A.37]{FoeSch}. 
Let $\var^\P$ be defined as in \eqref{eq:VaRintro}.
By \cite[Lemma A.32]{FoeSch}, 
\[\mathcal Z=\{\VaR^{\P}_{U}(Z)\colon Z\in\mathcal Z,\,U\in\mathcal U\},\]
and there is a particular $U^\star\in\mathcal U$ such that 
$Z^\star=\VaR^{\P}_{U^\star}(Z^\star)$.

As $1-U^\star\in\mathcal U$ as well, we have for all $Z\in\mathcal Z$ that $\VaR^{\P}_{1-U^\star}(Z)\in\mathcal Z$. 
Fix $s\in(0,1)$ and consider the  nontrivial event $A:=\{U^\star\le s\}$. 
As the function $(0,1)\ni\alpha\mapsto \VaR^{\P}_\alpha(X)$ associated with an arbitrary $X\in L^0_\P$ is  nondecreasing, we obtain
\begin{equation*}
\VaR^{\P}_{1-s}(Z)\ind_A\le \VaR^{\P}_{1-U^\star}(Z)\ind_A\le {Z^\star}\ind_A=\VaR^{\P}_{U^\star}(Z^\star)\ind_A\le \VaR^{\P}_s(Z^\star)\ind_A.
\end{equation*}
The latter estimate holds if and only if 
\begin{equation}\label{eq:first}\sup_{Z\in\mathcal Z}\VaR^{\P}_{1-s}(Z)\le \VaR^{\P}_s(Z^\star).\end{equation}
Taking the limit $s\downarrow 0$ in \eqref{eq:first}, we obtain $$a:=\sup_{Z\in\mathcal Z}\sup_{s\in(0,1)}\VaR^{\P}_{1-s}(Z)\le \inf_{s\in(0,1)}\VaR^{\P}_{s}(Z^\star)<\infty,$$ 
i.e., $a$ is a real constant.
Moreover, the preceding estimate yields for all $Z\in\mathcal Z$ and $U\in\mathcal U$ that
\[\VaR^{\P}_{U}(Z)\le a\le \var_{U^\star}^\P(Z^\star)=Z^\star.\]
i.e., the constant random variable $a$ is an upper bound of $\mathcal Z$ satisfying $a\le Z^\star$. By definition of a supremum, this is only possible if $a=Z^\star$.
\end{proof}

We can now proceed with the proof of Theorem~\ref{main1}.

\begin{proof}[Proof (of Theorem~\ref{main1})]
We only have to prove statement (i) as (ii)  follows by considering $-\ph$ instead of $\ph$. 

\textsc{Step 1:} 
If $\mathcal L$ is nonempty and $\sup\mathcal L$ exists in $\ca$, then it is a multiple of $\P$.
To prove this statement, we  first claim that each $\mu\in\mathcal L$ satisfies the absolute continuity relation $\mu\ll\P$; 
see Appendix~\ref{app:definitions} for the definition.  
Indeed, suppose an event $N\in\Sigma$ satisfies $\P[N]=0$ and let $X\in\CD$. For every $A\in\Sigma$ with $A\subseteq N$ and every $n\in\N$, $n\ind_A-n\ind_{N\setminus A}+X\ind_{N^c}\in\CD$. Hence, for arbitrary $\mu\in\mathcal L$, 
\begin{align}\label{eq:nullset}n|\mu|(N)+\E_{\mu}[X\ind_{N^c}]&=\sup\big\{\E_\mu[n\ind_A-n\ind_{N\setminus A}+X\ind_{N^c}]\colon A\in\Sigma,\,A\subseteq N\}\nonumber\\
&\le \ph(X)<\infty.\end{align}
Letting $n\to\infty$ implies that $|\mu|(N)=0$,
i.e.\ $\mu\ll \P$. 

Next, consider the lattice isomorphism   
$$\{\mu\in\ca\colon \mu\ll\P\}\to L^1_\P,\quad \mu\mapsto\tfrac{{\rm d}\mu}{{\rm d}\P}$$
produced by the Radon-Nikod\'ym derivative, which is a bijection between the two spaces and preserves order relations and operations.
For every $D=\tfrac{{\rm d}\mu}{{\rm d}\P}$, $\mu\in \mathcal L$, every $Z\sim_\P D$, and every $X\in\CD$, the Hardy-Littlewood bounds (\cite[Appendix A.3]{FoeSch}) deliver
$$\E_\P[ZX]\le \sup_{Z'\sim_\P D}\E_\P[Z' X]=\sup_{Y\sim_\P X}\E_\P[DY]=\sup_{Y\sim_\P X}\E_{\mu}[Y]\le\sup_{Y\sim_\P X}\ph(Y)= \ph(X).$$
Consequently, the signed measure defined by density $Z$ also lies in $\mathcal L$, and the set of densities $\mathcal Z:=\{\tfrac{{\rm d}\mu}{{\rm d}\P}\colon \mu\in\mathcal L\}$ is $\P$-invariant. 
Suppose now that $\mathcal L$ is bounded above in $\ca$.
By the lattice isomorphism property of the Radon-Nikod\'ym derivative, $\mathcal Z$ is also bounded above in $L^0_\P$ and its supremum is the density of $\sup\mathcal L$. 
Lemma~\ref{lem:constant} finally shows constancy of 
$\frac{{\rm d}\sup\mathcal L}{{\rm d}\P}$, or equivalently, $\sup\mathcal L$ being a multiple of $\P$.
    
\textsc{Step 2:} If $\mathcal L^f$ is nonempty, $\mathcal L$ is nonempty as well.  
As in \eqref{eq:nullset}, one observes for any $\mu\in\mathcal L^f$ that $\mu\ll\P$. 
By Lemma~\ref{lem:sup}, $\mu$ gives rise to a subadditive, positively homogeneous, continuous and $\P$-invariant functional 
\[\rho_{\mu}(X):=\sup_{Y\sim_\P X}\E_\mu[Y],\quad X\in B(\Sigma).\]
By \cite[Proposition 1.1]{Svindland}, there is a family $\mathcal M_\mu\subseteq \ca$ of signed measures $\zeta\ll\P$ such that 
\[\rho_\mu(X)=\sup_{\zeta\in\mathcal M_\mu}\E_\zeta[X],\quad X\in B(\Sigma).\]
In particular, for each $\zeta\in\mathcal M_\mu$ and $X\in\CD$,
    $\E_\zeta[X]\le \rho_\mu(X)\le \ph(X)$, meaning that 
\begin{equation}\label{eq:bigunion}\bigcup_{\mu\in\mathcal L^f}\mathcal M_\mu\subseteq \mathcal L.\end{equation}

\textsc{Step 3:} Existence of $\sup\mathcal L$ in $\ca$ implies that $\sup\mathcal L^f$ exists in $\ba$ and agrees with $\sup\mathcal L$.
By Step 1, $\sup\mathcal L=a\P$ for a suitable $a\in\R$. Given arbitrary $\mu\in\mathcal L^f$ and $A\in\Sigma$, 
\[\mu(A)\le \rho_\mu(\ind_A)=\sup_{\zeta\in\mathcal M_\mu}\zeta(A)\le a\P[A],\]
or that $\mu\le \sup\mathcal L$. We have used \eqref{eq:bigunion} in the last inequality. Hence, $\sup\mathcal L^f$ exists in $\ba$ and satisfies $\sup\mathcal L^f\le \sup\mathcal L$. 

\textsc{Step 4:} If $\sup\mathcal L^f$ exists in $\ba$, then $\sup\mathcal L$ exists in $\ca$ and satisfies $\sup\mathcal L\le\sup\mathcal L^f$.  
Indeed, as $\mathcal L\subseteq \mathcal L^f$, the existence of $\sup\mathcal L^f$ in $\ba$ implies that the set $\mathcal L$ is upper bounded in $\ba$.
    The latter is a Dedekind complete vector lattice (see Appendix~\ref{app:definitions}), thus $\sup\mathcal L$ exists in $\ba$ and satisfies $\sup\mathcal L\le \sup\mathcal L^f$. 
    As $\ca$ is a band in $\ba$ (see Appendix~\ref{app:definitions}), $\sup\mathcal L$ calculated in $\ba$ in fact lies in $\ca$. This means that the supremum of $\mathcal L$ also exists in the smaller space $\ca$.

    The proof is complete.
\end{proof}

We continue with the proof of Corollary~\ref{cor:remark}.

\begin{proof}[Proof (of Corollary~\ref{cor:remark})]
We aim to apply Theorem~\ref{main1}. 
To this effect, set $$\widetilde{\CD}:=\CD\cup\{X\in B(\Sigma)\colon X=-c~\P\text{-a.s.}\},$$ 
a $\P$-invariant subset of $B(\Sigma)$, and suppose that $\mathcal L$ (or equivalently, $\mathcal L^f$) is nonempty.  
Moreover, define $\widetilde\ph\colon\widetilde\CD\to\R$ by
$$\widetilde \ph(X)=\begin{cases}\ph(X)&\quad\text{if }\P[X\neq -c]>0,\\
-\ph(c)&\quad\text{if }X=-c~\P\text{-a.s.},\end{cases}$$
which satisfies
\begin{equation}\label{usedonce}\widetilde\ph|_{\mathcal D}\le \ph.\end{equation} 
This assertion is clear if $-c\notin\mathcal D$. 
Else, if $-c\in\mathcal D$ (or equivalently, $\CD=\widetilde\CD$), then $\mathcal L^f\neq\varnothing$ implies $\ph(-c)\ge -\ph(c)$, proving \eqref{usedonce}.

Denote the lower supporting sets of $\widetilde\ph$ by  $\widetilde{\mathcal L}$ and $\widetilde{\mathcal L}^f$.  
By \eqref{usedonce}, $\widetilde{\mathcal L}\subseteq\mathcal L$ and $\widetilde{\mathcal L}^f\subseteq\mathcal L^f$. 
Moreover, each $\mu\in\widetilde{\mathcal L}^f$ satisfies $\E_\mu[c]=\ph(c)$, i.e., $\widetilde{\mathcal L}\subseteq \mathcal L_c$ and $\widetilde{\mathcal L}^f\subseteq \mathcal L_c^f$. 
If $-c\notin\mathcal D$, the latter inclusion is an equality of sets. 
If $-c\in\CD$ and $\mu\in\mathcal L_c^f$ is arbitrary, then
\[\E_\mu[-c]=-\ph(c)\le \ph(-c).\]
This means that
$\widetilde{\mathcal L}=\mathcal L_c$ and $\widetilde{\mathcal L}^f=\mathcal L_c^f$, and it remains to apply Theorem~\ref{main1}. 
\end{proof}

\begin{proof}[Proof (of Theorem~\ref{main4})]
Arguing as in \eqref{eq:nullset}, each $\mu\in\mathcal L^f$ must satisfy $\mu\ll P$, giving rise to a functional $\rho_\mu$ on $B_s(\Sigma)$ by Lemma~\ref{lem:sup}. 
Also note that $\rho_\mu|_{\mathcal D}\le \ph$.
Defining the convex conjugate $\rho_\mu^*$ of $\rho_\mu$ on $\ba$ like in \eqref{def:conv conj}, 
we infer that
\[\mathcal L^f=\bigcup_{\mu\in\mathcal L^f}\{\rho_\mu^*<\infty\}.\]

Denote $\sup\mathcal L^f$ by $\mu^\star$ and fix $A\in\Sigma$ with $P[A]\in(0,1)$.
Using Lemma~\ref{lem:FH} and the notation therein, one obtains the following chain of estimates:
\begin{align*}
\mu^\star(\Omega)P[A]&\ge\iota_{\mu^\star}(A)\\
&\ge \inf_{P[B]=P[A]}\sup_{\mu\in\mathcal L^f}\mu(B)\\
&=\inf_{P[B]=P[A]}\sup_{\mu\in\mathcal L^f}\rho_\mu(\ind_B)\\
&=\sup_{\mu\in\mathcal L^f}\rho_\mu(\ind_A)\\
&\ge \sup_{\mu\in\mathcal L^f}\mu(A).
\end{align*}
This is sufficient to show that  $\mu^\star(\Omega)P$ is also an order upper bound of $\mathcal L^f$. 
By definition of the supremum, $\mu^\star(\Omega)P\ge \mu^\star$. 
This can only happen if $\mu^\star$ and $P$ are linearly dependent. 
\end{proof}

\section{Distortion riskmetrics}\label{sec:distortion}

An important class of law-invariant functionals are {\em distortion riskmetrics}, i.e., continuous functionals $\ph\colon B(\Sigma)\to\R$ that are comonotonic additive and $\P$-invariant for some probability measure $\P$.
The term ``distortion riskmetric'' is adopted from the recent work Wang et al.~\cite{Riskmetrics} to emphasise its possible lack of monotonicity, but the study of monotone distortion riskmetrics dates back at least to Quiggin \cite{Quiggin} and Yaari \cite{Yaari}  in economics.  
Comonotonic additive functionals play a central role not only in risk management, but also in economics and insurance; see, e.g., Acerbi \cite{Acerbi} Gilboa \cite{Gilboa}, Kou and Peng \cite{KP16}, Schmeider~\cite{Schmeidler}, and Wang et al.~\cite{Insurance}.  

This section is motivated by a simple observation.
Comonotonic additivity implies that a continuous functional is fully characterised by its values on the set
\[
\CI := \{\ind_A \colon A \in \Sigma\}
\]
of indicator functions. In other words, the restriction $\ph|_{\CI}$ contains all essential information. 
Equivalently, this restriction can be identified with a set function $v$, often called a {\em cooperative game}, obtained by $v(A):=\ph(\ind_A)$ for $A\in\Sigma$.
The associated distortion riskmetric is then recovered from $v$ via the Choquet integral; see Marinacci and Montrucchio~\cite{MariMont} and Schmeidler \cite{Schmeidler}.
From a risk measurement perspective, this reduction is useful because the risk of arbitrarily complex random losses can be derived from the risk profile of simple binary losses---those delivering a unit loss on a specific loss event and zero loss on its complement.
In the remainder of this section, we shall therefore study games instead of functionals.

Crucially, the set $\mathcal I$ is {\em not} $\P$-invariant because nontrivial null events can exist.
For instance, if we can pick $A,B\in\Sigma$ pairwise disjoint with $\P[B]=0$, then $\ind_A\sim_\P \ind_A+x\ind_B$ for every $x\in\R$,
but the latter random variable might not be an element of $\mathcal I$. 
Hence, the results from Sect.~\ref{sec:first} are not immediately applicable to $\ph|_{\CI}$.

This problem could be overcome by considering $\ph$ on the larger domain $$\widetilde \CI:=\{X\colon \P[X\in\{0,1\}]=1\},$$
which is invariant with respect to each $\Q$ equivalent to $\P$. 
However, taking seriously the concern raised in Sect.~\ref{sec:first} about {\em a priori} knowledge of the reference probability, one is led to ask whether the analysis carries over to the completely ``model-free'' setting of $\mathcal I$ in the important special case of distortion riskmetrics.
The aim of the present section is to provide an affirmative answer.

\subsection{Versions of the main results for games}

While the link to cooperative game theory will not play a substantial role,  we shall use game-theoretic terminology in order to be concise.
Consider the setting of Assumption~\ref{ass:structure}(B), i.e., $\Sigma$ is an algebra and $P$ is a convex-ranged probability charge. A set function $v\colon \Sigma\to\R$ is called:
\begin{itemize}
    \item {\em $P$-invariant} if $v(A)=v(B)$ whenever $A,B\in\Sigma$ satisfy $P[A]=P[B]$;
    \item {\em monotone} if $v(A)\le v(B)$ whenever $A,B\in\Sigma$ satisfy $A\subseteq B$;
    \item {\em continuous at $\varnothing$} if $\lim_{n\to\infty}v(A_n)=v(\varnothing)$ whenever $A_n\downarrow \varnothing$;
    \item a {\em game} if $v(\varnothing)=0$. 
\end{itemize}
If $v$ is $P$-invariant, then there exists a unique function $h\colon[0,1]\to\R$ such that $v=h\circ P$. Indeed, the convex range of $P$ ensures that for every $p\in[0,1]$ there exists $A_p\in\Sigma$ with $P(A_p)=p$, and the function $h$ is recovered via $h(p)=v(A_p)$.
The function $h$ is called a {\em distortion} of $P$, and if $v=h\circ P$ is a game, the corresponding distortion riskmetric is given by the Choquet integral with respect to $h\circ P$ in case the latter is continuous.

{\em Core} and {\em $\sigma$-core} of a game $v$ are the sets 
\begin{equation}\label{eq:core1}\core_v:=\{\mu\in\ba\colon \mu\ge v,\,\mu(\Omega)=v(\Omega)\}\quad\text{and}\quad\core^\sigma_v:=\core_v\cap\ca.\end{equation}
The  
{\em loose core} and the {\em loose $\sigma$-core} of $v$ are defined as 
\begin{equation}\label{eq:core2}
\LC_v:=\{\mu\in\ba\colon \mu\ge v\}\quad\text{and}\quad\LC^\sigma_v:=\LC_v\cap\ca,\end{equation}
omitting the normalisation constraint $\mu(\Omega)=v(\Omega)$.
The defining inequalities in \eqref{eq:core1} and \eqref{eq:core2} are understood setwise, and the defined sets may well be empty.
{\em Anticore} $\acore_v$, {\em $\sigma$-anticore} $\acore^\sigma_v$, {\em loose anticore} $\LA_v$, and {\em loose $\sigma$-anticore} $\LA^\sigma_v$ of $v$ are defined by replacing the condition $\mu\ge v$ in \eqref{eq:core1} and \eqref{eq:core2} by $\mu\le v$. 

The following proposition is an analogue of Theorem~\ref{main1} and Corollary~\ref{cor:remark} for games.

\begin{proposition}\label{main3}
Suppose $\Sigma$ is a $\sigma$-algebra, $\P$ is an atomless probability measure, and $v$ is a $\P$-invariant game. 
\begin{enumerate}
\item[(i)] For each choice $\mathcal Y\in\{\acore_v,\LA_v\}$, 
we have that $\sup\mathcal Y$ exists in $\ba$ if and only if $\sup\mathcal Y^\sigma$ exists in $\ca$.
In that case, there is a constant $a\in\R$ such that 
$$a\P=\sup\mathcal Y=\sup\mathcal Y^\sigma.$$
\item[(ii)] For each choice $\mathcal Y\in\{\core_v,\LC_v\}$, 
we have that $\inf\mathcal Y$ exists in $\ba$ if and only if $\inf\mathcal Y^\sigma$ exists in $\ca$.
In that case, there is a constant $b\in\R$ such that 
$$b\P=\inf\mathcal Y=\inf\mathcal Y^\sigma.$$
\end{enumerate}
\end{proposition}

We also have a direct counterpart to Theorem~\ref{main4} for the more general case of a finitely additive reference probability: 

\begin{proposition}\label{cor:main1}
Let $P$ be a convex-ranged probability charge on an algebra $\Sigma$, and let $v\colon \Sigma\to\R$ be a $P$-invariant game.
\begin{enumerate}
\item[(i)] Let $\CY\in\{\acore_v,\LA_v\}$ and suppose 
$\sup\CY$ exists. Then there is a constant $a\in\R$ such that 
\[\sup\CY=a P.\]
\item[(ii)] Let $\CY\in\{\core_v,\LC_v\}$ and suppose 
$\inf\CY$ exists. Then there is a constant $b\in\R$ such that 
\[\inf\CY=b P.\]
\end{enumerate}
\end{proposition}

We prove Proposition~\ref{cor:main1} first, as the statement will be used in the proof of Proposition~\ref{main3}.

\begin{proof}[Proof (of Proposition~\ref{cor:main1})]
    Again, we only prove statement (i) and focus on the loose anticore $\LA_v$ first. Fix $\mu\in\LA_v$ and let $N\in\Sigma$ with $P[N]=0$. \cite[Theorem 10.53]{Ali} shows for the positive part $\mu^+=\mu\vee 0$ of $\mu$ that
    $$\mu^+(N)=\sup\{\mu(A)\colon A\in 
    \Sigma,\,A\subseteq N\}\le \sup\{v(A)\colon A\in 
    \Sigma,\,A\subseteq N\}=v(\varnothing)=0;$$
    i.e., $\mu^+\ll P$ holds. This allows to invoke Lemma~\ref{lem:concave} to obtain the $P$-invariant exact game of bounded variation $s_\mu$. In particular, 
    \begin{equation}\label{eq:union}\LA_v=\bigcup_{\mu\in\LA_v}\LA_{s_\mu}.\end{equation}
    Denote $\sup\LA_v$ by $\mu^\star$.
    Let $(\mathcal F_\alpha)$ be a net of finite subsets of $\LA_v$ such that $\mu_\alpha:=\max_{\mu\in \CF_\alpha}\mu$ increases to $\mu^\star$ in order. 
    By \cite[Lemma 8.15]{Ali}, $\mu_\alpha^+\uparrow (\mu^\star)^+$ in order, as well. 
    As each $\mu_\alpha$ satisfies $\mu_\alpha^+\ll P$, we also have $(\mu^\star)^+\ll P$. 
    
    Using Lemma~\ref{lem:concave} for the first inequality and \eqref{eq:union} for the first equality, one obtains for all $A\in\Sigma$ the following chain of estimates:
\begin{align*}
\mu^\star(\Omega)P[A]&\ge\iota_{\mu^\star}(A)\\
&\ge \inf_{P[B]=P[A]}\sup_{\mu\in\LA_v}\mu(B)\\
&=\inf_{P[B]=P[A]}\sup_{\mu\in\LA_v}s_\mu(B)\\
&=\sup_{\mu\in\LA_v}s_\mu(A)\\
&\ge \sup_{\mu\in\LA_v}\mu(A).
\end{align*}
In summary, $\mu^\star(\Omega)P$ is also an upper bound of $\LA_v$. 
By definition of the supremum, $\mu^\star(\Omega)P\ge \mu^\star$, which can only happen if $\mu^\star=\mu^\star(\Omega)P$. 
The case $\CY=\acore_v$ is treated analogously.
\end{proof}

We continue with the proof of Proposition~\ref{main3}.

\begin{proof}[Proof (of Proposition~\ref{main3})]
    In the case of $\acore_v$ and $\acore^\sigma_v$, each $\mu$ in these sets will satisfy $\mu\ll\P$ by construction. 
    The proof in these cases thus works by applying Theorem~\ref{main1} to the set
    \[
    \mathcal D=\{X\in B(\Sigma)\colon X=\ind_A~\P\text{-a.s.\ for some }A\in\Sigma\text{ or }X=-1~\P\text{-a.s.}\}.
    \]
    Hence, we shall focus on the ($\sigma$-)loose anticore. 

    Suppose we can pick $\mu\in\LA_v$. Necessarily, $\mu^+\ll \P$ and $s_\mu$ defined in the context of Lemma~\ref{lem:concave} is a bounded submodular $\P$-invariant game. 
    The associated signed Choquet integral $\ph_\mu$ is sublinear, $\P$-law invariant and continuous; see Wang et al.~\cite{WWW}. 
By \cite[Proposition 1.1]{Svindland}, there must be a subset $\mathcal M_\mu\subseteq\ca$ such that $\zeta\ll\P$ for all $\zeta\in\mathcal M_\mu$, and 
$$\ph_\mu(X)=\sup_{\zeta\in\mathcal M_\mu}\E_\zeta[X],\quad X\in B(\Sigma).$$
Hence, 
\[v(A)\ge s_\mu(A)=\ph_\mu(\ind_A)=\sup_{\zeta\in\mathcal M_\mu}\zeta(A),\]
meaning that 
\begin{equation}\label{eq:union2}
\bigcup_{\mu\in\LA_v}\mathcal M_\mu\subseteq \LA_v^\sigma,\end{equation}
the latter set being nonempty as well. Hence, if $\sup\LA_v$ exists in $\ba$, we can argue like in Step 4 of the proof of Theorem~\ref{main1} to show the existence of $\sup\LA_v^\sigma$ in $\ca$ and that $\sup\LA_v^\sigma\le\sup\LA_v$.

Suppose now that $\sup\LA_v^\sigma$ exists in $\ca$. 
For all $\nu\in\LA_v$ and all $A\in\Sigma$, we infer from \eqref{eq:union2} that
$$\nu(A)\le \sup_{\mu\in\LA_v}s_\mu(A)=\sup_{\mu\in\LA_v}\sup_{\zeta\in\mathcal M_\mu}\zeta(A)\le \sup_{\mu\in\LA_v^\sigma}\mu(A)\le (\sup\LA_v^\sigma)(A).$$
Hence, $\LA_v$ is upper bounded by $\sup\LA_v^\sigma$, meaning that $\sup\LA_v\le\sup\LA_v^\sigma$. 

The proof concludes with Proposition~\ref{cor:main1}.
\end{proof}

Several caveats are worth noting. The
 first concerns the observation that the sets $\acore_v^{(\sigma)}$, $\LA_v^{(\sigma)}$, $\core_v^{(\sigma)}$, and $\LC_v^{(\sigma)}$ are not always equally suited to elicit $\P$. 
 In fact, the existence of a nontrivial supremum/infimum of the loose ($\sigma$-)anticore/($\sigma$-)core has strong consequences for the regularity of $v$ under mild conditions. 

\begin{proposition}\label{prop:exist}
    Suppose $\Sigma$ is an algebra, $P$ is a convex-ranged probability charge, and $v=h\circ P$ is a nonnegative $P$-invariant game satisfying
    \begin{equation}\label{eq:inf}\inf_{0<x\le 1}h(x)>0.\end{equation}
    Then 
    $\core_v=\varnothing$
and $\sup\LA_v$ does not exist.
\end{proposition}

\begin{proof}
For the assertion $\core_v=\varnothing$, suppose that $\core_v$ is nonempty. 
Following the proof of \cite[Lemma 3]{Basket1}, we get $h(\frac 1 n)\le \frac{h(1)}n$ for all $n\ge 2$. 
This contradicts \eqref{eq:inf}.

Now assume towards a contradiction that $\sup\LA_v$ exists. By Proposition~\ref{cor:main1}, there is some $a\in\R$ such that $\sup\LA_v=aP$.
By \eqref{eq:inf}, we find $\delta>0$ such that $\delta\ind_{(0,1]}\le h$.
For each $A\in\Sigma$ with $P[A]>0$, we thus have $\delta P^A\in \LA_v$. 
Hence, 
$a P[A]\ge \delta$ must hold for all such $A\in\Sigma$. Letting $P[A]\downarrow 0$ yields a contradiction.  
\end{proof}

\begin{remark}\,
\begin{enumerate}
    \item[(a)] Suppose that the game $v=h\circ \P$ is induced by a distortion riskmetric. In this case, condition \eqref{eq:inf} reflects a very conservative risk management approach in which speculative losses occurring with very small probability carry nontrivial marginal risk. 
    In particular, $v$ often fails to be continuous at $\varnothing$. 
    This is typically not reasonable, and in such a context, one would expect \eqref{eq:inf} to fail.
    \item[(b)] Given a nonnegative $P$-invariant game $v=h\circ P$, reasoning symmetrically to the previous proposition yields that the condition 
    $\sup_{0\le x<1}h(x)<h(1)$ implies $\acore_v=\varnothing$.
    In the case of monotone games, this is tantamount to discontinuity of $h$ at $1$.
    Consequently, one should try eliciting $P$ in these cases by focusing on $\inf\LC_v$ or, if the reference probability is countably additive, on $\inf\LC_v^\sigma$.
\end{enumerate}
\end{remark}

A second caveat concerns the computation of suprema and infima in Theorem~\ref{main1}, as well as Proposition~\ref{main3}. 
These are taken in the spaces $\ca$ or $\ba$, the latter being the norm dual space of $B(\Sigma)$ and $B_s(\Sigma)$, the spaces of bounded and simple random variables, respectively.
Equivalently, infima and suprema can also be taken in the space $\ba_\P$ of all $\mu\in\ba$ with $\mu\ll \P$.
$\ba_\P$ is the dual space of $L^\infty_\P$, the space of equivalence classes of bounded random variables up to $\P$-a.s.\ equality. 

Do infima or suprema in one of the primal function spaces $B(\Sigma)$, $B_s(\Sigma)$, or $L^\infty_\P$, also contain sufficient relevant information about the reference probability to elicit the latter?
The order properties of $L^\infty_\P$ are particularly appealing. 
This space is {\em super Dedekind complete}, meaning every upper bounded subset has a supremum that is attained by a {\em countable} subset of the original set (\cite[Theorem A.37]{FoeSch}). The spaces $B_s(\Sigma)$ and $B(\Sigma)$, in contrast, do not admit suprema for every upper bounded subset.
Our natural question has a negative answer, however, see Lemma~\ref{lem:constant}.
The supremum $\sup\mathcal Z$ of a $\P$-invariant set $\mathcal Z\subseteq L^0_\P$ does not depend on the reference probability whatsoever and therefore contains no relevant information about the latter.

\subsection{The elicitation procedure as sandwich theorem}

The underlying mathematical structure of our elicitation results may still appear opaque to the reader. 
This subsection is therefore devoted to a more detailed examination in the special case of sub-/superadditive games.
We will see that our results can then be cast as sandwich theorems---a type of separation result---in the spirit of Kindler~\cite{Kindler}.
We also refer to the more operational approach to sandwich theorems of Amarante~\cite{Sandwich}, formulated in a very similar setting.

A set function $v\colon\Sigma\to[-\infty,\infty)$ is {\em superadditive} if for all pairwise disjoint events $A,B\in\Sigma$, 
\[v(A\cup B)\ge v(A)+v(B).\]
A set function $v\colon\Sigma\to(-\infty,\infty]$ is subadditive if $-v$ is superadditive. 
For a nonempty set $\mathcal R\subseteq \ba$ of signed charges, we define the {\em lower envelope} 
$\lo(\mathcal R):=\inf_{\mu\in\mathcal R}\mu(\cdot)$, which 
is superadditive, and the {\em upper envelope} $\up(\mathcal R):=\sup_{\mu\in\mathcal R}\mu(\cdot)$, which is subadditive. 
In the situation of Proposition~\ref{cor:main1}, nonemptiness of $\LA_v$ implies that $\up(\LA_v)$ is in fact a subadditive game and that $\up(\LA_v)\le v$ holds setwise.
Moreover, if $\sup\LA_v$ exists in $\ba$, then also 
\begin{equation}\label{eq:sandwich1}\up(\LA_v)\le \sup\LA_v
\end{equation}
holds setwise.
While we have already remarked in Sect.~\ref{sec:first} that there need not be a setwise order relationship between $v$ and $\sup\LA_v$, suppose for the moment that, in addition to \eqref{eq:sandwich1}, we have 
\begin{equation}\label{eq:sandwich2}\up(\LA_v)\le \sup\LA_v\le v.\end{equation}
In that case, the linear object $\sup\LA_v\in\ba$ is sandwiched between the subadditive game $\up(\LA_v)$ and the game $v$, thus providing a linear separation between the two. 

Sandwich theorems give sufficient conditions for two set functions $v\le w$ to admit
a linear separation as in \eqref{eq:sandwich2}, i.e., for the existence of
$\mu\in\ba$ with $v\le \mu\le w$.
Proposition~\ref{prop:super}, the version of our elicitation result for superadditive games, can be viewed as a sandwich theorem. It characterises $\sup\LA_v$ as a sandwiched functional and leverages the special focus on $P$-invariance to obtain a more precise separation than, e.g., \cite[Proposition 3]{Sandwich}.
Not least, the preceding discussion immediately transfers to $\acore_v$ as well as $\core_v$ and $\LC_v$---replacing supremum by infimum and flipping inequalities.

\begin{proposition}\label{prop:super}
Let $P$ be a convex-ranged probability charge on an algebra $\Sigma$ and suppose that $v$ is a $P$-invariant superadditive game. Then the following statements are equivalent:
\begin{enumerate}
    \item[(i)] $\LA_v\neq\varnothing$.
    \item[(ii)] $\sup\LA_v$ exists in $\ba$.
    \item[(iii)] $\{a\in\R\colon aP\le v\}\neq\varnothing$.
    \item[(iv)] $a^\star:=\sup\{a\in\R\colon a P\le v\}$ is finite.
\end{enumerate}
In this case, 
\begin{equation}\label{eq:implication1}\LA_v=\{\mu\in\ba\colon \mu\leq a^\star P\}\qquad\text{and}\qquad\sup\LA_v=a^\star P.\end{equation}
{\em A fortiori}, $a^\star P$ is the unique element of $\ba$ sandwiched between the subadditive game $\up(\LA_v)$ and $v$ itself. 
\end{proposition}

\begin{proof}
We first prove the equivalence of statements (i)--(iv).

(i) implies (iii): Let $\mu\in\LA_v$. Either $\mu$ itself is a multiple of $P$, or the associated game $s_\mu$ satisfies $\mu(\Omega)P\le s_\mu\le v$, i.e., $\mu(\Omega)\in\{a\in\R\colon aP\le v\}$ is nonempty. 

(iii) implies (iv): Whenever $b>v(\Omega)$, $bP[\Omega]=b>v(\Omega)$. Thus, 
$$\{a\in\R\colon aP\le v\}\subseteq(-\infty,v(\Omega)]$$
and $\sup\{a\in\R\colon aP\le v\}$ is a real number. 

(iv) implies (ii): First, let $\mu,\nu\in \LA_v$. By~\cite[Theorem 10.53]{Ali}, their maximum $\mu\vee\nu$ can be computed at $A\in\Sigma$ as
\begin{align*}
(\mu\vee\nu)(A)&=\sup\{\mu(B)+\nu(A\setminus B)\colon B\in\Sigma,\,A\supseteq B\}\\
&\le \sup\{v(B)+v(A\setminus B)\colon B\in\Sigma,\,A\supseteq B\}\\
&\le v(A).   
\end{align*}
In the first inequality we have used that $\mu,\nu\le v$, in the second that $v$ is superadditive. Consequently, $\mu\vee\nu\in\LA_v$ again.

Select $\mu\in\LA_v$ arbitrarily. 
The set $\mathcal S:=\{\nu\in \LA_v\colon \nu\ge \mu\}$ forms a nondecreasing net in $\ba$. Moreover, for all $\nu\in\mathcal S$, 
the total variation norm $TV(\nu^+)$ of $\nu^+$ satisfies
$$TV(\nu^+)=\nu^+(\Omega)=\sup_{A\in\Sigma}\nu(A)\le\sup_{A\in\Sigma}v(A).$$
Using superadditivity of $v$ and statement (iv), 
\begin{align*}
\sup_{A \in \Sigma} v(A) 
&\le \sup_{A \in \Sigma} \big\{ v(\Omega) - v(A) \big\} = v(\Omega) - \inf_{A \in \Sigma} v(A) \\
&\le v(\Omega) - \inf_{A \in \Sigma} a^\star P[A] \le v(\Omega) + |a^\star|.
\end{align*}
Hence, the set $\mathcal S$ is norm bounded in $\ba$. As $\ba$ is {\em monotonically complete} in the sense of \cite[Definition 2.4.18]{MeyNie} according to \cite[Proposition 2.4.19(ii)]{MeyNie}, $\sup\mathcal S=\sup\LA_v$ exists in $\ba$. 

(ii) implies (i) by definition. 

In order to verify \eqref{eq:implication1}, suppose that $\sup\LA_v$ exists. 
Proposition~\ref{cor:main1} proves that $\sup\LA_v=cP$ for a suitable $c\in\R$. 
The property of $\LA_v$ being directed upwards shows for all $A\in\Sigma$ that
\begin{equation}\label{eq:c}
c P[A]=\sup_{\mu\in\LA_v}\mu(A)\le v(A),\end{equation}
i.e.\ $c\le a^\star$. In view of $a^\star P\in \LA_v$, we have $c=a^\star$.
Another consequence of \eqref{eq:c} is that $\LA_v=\{\mu\in\ba\colon \mu\le a^\star P\}$.
This suffices for the sandwich property of $a^\star P$. 
\end{proof}

Corollary~\ref{cor:sub} is the immediate mirror image of Proposition~\ref{prop:super}. 

\begin{corollary}\label{cor:sub}
Let $P$ be a convex-ranged probability charge on an algebra $\Sigma$ and suppose that $v$ is a $P$-invariant subadditive game. Then the following statements are equivalent:
\begin{enumerate}
    \item[(i)] $\LC_v\neq\varnothing$.
    \item[(ii)] $\inf\LC_v$ exists.
    \item[(iii)] $\{b\in\R\colon bP\ge v\}\neq\varnothing$.
    \item[(iv)] $b^\star:=\inf\{b\in\R\colon b P\ge v\}$ is a real number.
\end{enumerate}
In this case, 
$$\LC_v=\{\mu\in\ba\colon \mu\ge b^\star P\}\qquad\text{and}\qquad\inf\LC_v=b^\star P.$$ 
{\em A fortiori}, $-b^\star P$ is the unique element of $\ba$ sandwiched between the subadditive game $-\lo(\LC_v)$ and the superadditive game $-v$. 
\end{corollary}

\section{Examples}
\label{sec:ex}

We now illustrate our results with several prominent examples from the literature on risk measures. The entropic risk measures, Expected Shortfall, and Value-at-Risk---discussed in detail in~\cite[Chapter 4]{FoeSch} and~\cite[Chapter 2.3]{MFE15}---are arguably the three most popular one-parameter families.
Throughout, we work in the framework of Assumption~\ref{ass:structure}(B).

\subsection{Entropic risk measure}\label{ex:entr}

Consider the class of entropic risk measures $\Entr_\alpha^P$, $\alpha>0$ being a fixed parameter, defined by the formula  
\[\Entr_\alpha^P(X)=\tfrac 1 \alpha\log(\E_P[e^{\alpha X}]),\quad X\in B_s(\Sigma).\]
If $\Sigma$ is a $\sigma$-algebra, this functional can be defined on the larger space $B(\Sigma)$ without any problem. 
We claim that  \begin{equation}\label{eq:all}\mathcal L^f=\{P\}\quad\text{and}\quad\mathcal U^f=\varnothing,\end{equation}
meaning that $\sup\mathcal L^f=P$ while $\inf\mathcal U^f$ does not exist. 
In order to verify \eqref{eq:all}, consider $X\in B_s(\Sigma)$ which can be represented as $X=\sum_{i=1}^nx_i\ind_{A_i}$ for suitable pairwise disjoint events $A_1,...,A_n$. As
\[\Entr_\alpha^P(X)=\tfrac 1 \alpha\log\bigg(\sum_{i=1}^ne^{\alpha x_i}P[A_i]\bigg),\]
$P\in\mathcal L^f$ follows from the finite-dimensional Jensen  inequality. Now, for arbitrary $\mu\in\mathcal L^f$, $X$ as above, and $t>0$, 
\begin{equation}\label{eq:Gateaux}\E_\mu[X]\le\frac{\Entr_\alpha^P(tX)}t=\frac{\log(\sum_{i=1}^ne^{\alpha tx_i}P[A_i])}{\alpha t}\to \E_P[X],\quad t\downarrow 0.\end{equation}
As $X$ was arbitrary, this can only hold if $\mu=P$. 
Symmetrically, if we could choose $\mu\in\mathcal U$, we would have for all $A\in\Sigma$ with $P[A]>0$ that 
\begin{equation}
\label{eq:no}\mu(A)\ge \lim_{t\uparrow\infty}\frac{\log(e^{\alpha t}\P[A]+1-P[A])}{\alpha t}=\lim_{t\uparrow\infty}\frac{e^{t}\P[A]}{e^{t}\P[A]+1-P[A]}=1.\end{equation}
As $P$ has convex range, there are arbitrarily large finite partitions of $\Omega$ into events with positive $P$-probability. Hence, no $\mu\in\ba$ can satisfy inequality \eqref{eq:no}.

Equation~\eqref{eq:Gateaux} shows that $P$, the unique element of $\mathcal L^f$, may be interpreted as the directional derivative
$\lim_{t \downarrow 0} \frac{\Entr_\alpha^P(tX)}{t} = \E_P[X]$.
Hence, access to the values of the entropic risk measure on $B_s(\Sigma)$ enables the decision maker to compute this derivative and thereby identify the reference probability $P$, independently of $\alpha > 0$. 
Nevertheless, the inability to infer $\alpha$ itself may be regarded as a limitation. 
However, we shall see momentarily that the dual infima and suprema respond sensitively to the domain on which the functional is defined, which will allow to fix this problem.

The entropic risk measure is not a distortion riskmetric. Nevertheless, it induces a game ${v_\alpha}(A):=\Entr_\alpha^P(\ind_A)$ which is given by applying the concave transformation
\begin{equation}\label{eq:h}
h_\alpha(x)=\tfrac 1 \alpha\log\big((e^\alpha-1)x+1\big),\quad x\ge 0,\end{equation}
to the argument $P[A]$.
Examining this capacity instead of the functional on the whole space provides a markedly different perspective.
Indeed,
\begin{equation}
    \label{claimENT}
\sup\acore_{v_\alpha}=\sup\LA_{v_\alpha}=\inf\LC_{v_\alpha}=\frac{e^\alpha-1}\alpha P,
\end{equation}
i.e., not only do the infimum of the loose core and the supremum of the (loose) anticore agree, but they also allow us to elicit {\em both} the reference probability $P$ and the parameter $\alpha$. 

To prove \eqref{claimENT}, we focus first on the loose core $\LC_{v_\alpha}$.
Concavity of the function $h_\alpha$ in \eqref{eq:h} implies that the constant $b^\star:=\inf\{b\in\R\colon bP\ge {v_\alpha}\}$ is given by the right-hand derivative $h_\alpha'(0)=\frac{e^\alpha-1}\alpha$. 
In view of Corollary~\ref{cor:sub}, subadditivity of ${v_\alpha}$ implies 
$\inf\LC_{v_\alpha}=\frac{e^\alpha-1}\alpha P.$
Regarding $\acore_{v_\alpha}$ and $\LA_{v_\alpha}$, these sets are order bounded above by $\inf\LC_{v_\alpha}$ and thus have a supremum. 
For events $D$ with $0<P[D]<1$, 
we set
$$\mu_D:=h_\alpha(P[D])P^D+\big[1-h_\alpha(P[D])\big]P^{D^c}.$$ 
It can be shown that $\mu_D$ is a charge in $\acore_{v_\alpha}$. 

Now, for $A\in\Sigma$ with $p:=P[A]>0$, let $D_1,\dots,D_n$ form a measurable partition of $A$ such that each event has probability $p/n$. 
We estimate
\begin{align*}(\sup\LA_{v_\alpha})(A)&\ge (\sup\acore_{v_\alpha})(A)=\sum_{i=1}^n(\sup\acore_{v_\alpha})(D_i)\\
&\ge \sum_{i=1}^n\mu_{D_i}(D_i)=nh_\alpha\big(\tfrac{p}n\big)=\frac{h_\alpha(p/n)}{p/n}p.
\end{align*}
Letting $n\to\infty$ delivers 
$$(\sup\LA_{v_\alpha})(A)\ge (\sup\acore_{v_\alpha})(A)\ge h_\alpha'(0)P[A]=(\inf\LC_{v_\alpha})(A),$$
which is sufficient for the claim.

\subsection{Expected Shortfall}\label{ex:ES}

For a parameter $\beta \in [0,1)$, the Expected Shortfall (ES) risk measure at level $\beta$ is defined by
\[
\es_\beta^P(X) := \frac{1}{1 - \beta} \int_\beta^1 \var_q^P(X) \, {\rm d}q, \quad X \in B_s(\Sigma).
\]
Here,
\begin{equation}\label{eq:VaR}
\VaR^P_q(X) := \inf\{x \in \mathbb{R} \colon P[X \le x] \ge q\}
\end{equation}
denotes the Value-at-Risk at level $q$ under $P$. It is well known that $\es^P_\beta$ is a distortion riskmetric with associated subadditive capacity
\begin{equation}\label{ES cap}
v_\beta(A) := \min\big\{ \tfrac{P[A]}{1 - \beta}, 1 \big\}, \quad A \in \Sigma.
\end{equation}
As $v_\beta$ determines the functional on $B_s(\Sigma)$ uniquely, restricting one's attention to its smaller domain of indicator random variables is readily justified.

We first consider $\es_\beta^P$ defined on the entire space $B_s(\Sigma)$. 
If $\beta\in(0,1)$, the upper supporting set $\mathcal U^f$ is empty because it could only contain probability charges, but the condition $\E_\mu[\cdot]\ge\es_\beta$ forces $\mu(\Omega)\ge \frac 1{1-\beta}$.
The lower supporting set $\mathcal{L}^f$ coincides with the anticore $\acore_{v_\beta}$. From \eqref{eq:x} below, we obtain
\[\sup \mathcal{L}^f = \tfrac{1}{1 - \beta} P,
\]
i.e., at least the dual supremum exists. 
Moreover, in contrast to the entropic risk measure case in Sect.~\ref{ex:entr}, we can elicit both the reference probability $P$ and the level $\beta$.

Now, consider the capacity $v_\beta$. We claim that
\begin{equation}\label{eq:x}
\sup \acore_{v_\beta} = \sup \LA_{v_\beta} = \inf \LC_{v_\beta} = \tfrac{1}{1 - \beta} P.
\end{equation}
The conclusion regarding the loose core follows analogously to the reasoning in Sect.~\ref{ex:entr}. For the analysis of the (loose) anticore, we introduce the notation
\[
\mu_D := \tfrac{P[D]}{1 - \beta} P^D + \tfrac{1 - \beta - P[D]}{1 - \beta} P^{D^c}
\]
for events $D$ satisfying $0 < P[D] < 1 - \beta$. Each charge $\mu_D$ lies in $\acore_{v_\beta}$.

Now, fix an event $A \in \Sigma$ with $P[A] > 0$, and partition $A$ into disjoint events $D_1, \dots, D_n$, each with $P[D_i] < 1 - \beta$. Then,
\[
(\sup \LA_{v_\beta})(A) \ge (\sup \acore_{v_\beta})(A) = \sum_{i=1}^n (\sup \acore_{v_\beta})(D_i) \ge \sum_{i=1}^n \mu_{D_i}(D_i) = \frac{P[A]}{1 - \beta}.
\]
This suffices to establish the claim in \eqref{eq:x}.

\subsection{Value-at-Risk}\label{ex:VaR}

Now we consider the Value-at-Risk---or quantile---class $\VaR^P_\gamma$, $0<\gamma<1$, defined by \eqref{eq:VaR}. We exclude the degenerate cases $\var^P_0$ and $\var^P_1$ because they are not of practical or regulatory relevance and do not have a unique reference probability, see Liebrich~\cite{Liebrich}.
For the associated game ${v_\gamma}\colon \Sigma\to\R$ given by 
\begin{equation}\label{eq:vgamma}v_\gamma(A)=\var^P_\gamma(\ind_A)=\begin{cases}1&\quad\text{if }P[A]>1-\gamma,\\
0&\quad\text{else},\end{cases}\end{equation}
and its loose (anti)core, we claim that
\begin{equation}\label{sup1}\sup\LA_{v_\gamma}=0\end{equation}
and
\begin{equation}\label{inf1}\inf\LC_{v_\gamma}=0.\end{equation}
For \eqref{sup1}, fix arbitrary $\mu\in\LA_{v_\gamma}$ and $A\in\Sigma$ and split $A$ into subevents $D_1,\dots,D_n$ of $P$-probability at most $1-\gamma$ to obtain 
\[\mu(A)=\sum_{i=1}^n\mu(D_i)\le \sum_{i=1}^n{v_\gamma}(D_i)=0.\]
As for \eqref{inf1}, we first observe that $(1-\gamma)^{-1} P\in\LC_{v_\gamma}$ and that each $\mu\in\mathcal\LC_{v_\gamma}$ satisfies $\mu\ge 0$.  
Now let $D\in\Sigma$ with $\delta:=P[D]\in(0,1-\gamma)$. 
For arbitrary $s\in (0,\frac 1 2)$ choose 
$$x_s:=\max\Big\{\tfrac{(1-s)(1-\delta)}{1-\gamma-\delta},\tfrac{(1-\delta)s}{\delta}+1\Big\}$$
and consider 
$$\mu_{D,s}:=s P^D+x_sP^{D^c}.$$
For $A\in\Sigma$ with $P[A]>1-\gamma$, we infer from $s/\delta<x_s/(1-\delta)$ that
$$\mu_{D,s}(A)=\tfrac s\delta P[A\cap D]+\tfrac{x_s}{1-\delta}P[A\cap D^c]>s+\tfrac{x_s(1-\gamma-\delta)}{1-\delta}\ge 1.$$
This means that $\mu_{D,s}\in\LC_{v_\gamma}$ and that 
$$(\inf\LC_{v_\gamma})(D)\le \inf_{0<s<\frac 1 2}\mu_{D,s}(D)=0.$$
\eqref{inf1} is now established by partitioning arbitrary events $A$ into finitely many pieces whose $P$-probability is in $(0,1-\gamma)$ and using additivity of $\inf\LC_{v_\gamma}$.

In summary, while Proposition~\ref{cor:main1} holds true, we fail to elicit the reference probability of the VaR-capacity.

\section{Value-at-Risk}\label{sec:VaR}

We have seen in Sect.~\ref{ex:VaR} that our approach fails when applied to the Value-at-Risk capacity. 
However, Liebrich~\cite{Liebrich} shows that the VaR admits only one convex-ranged reference probability.
In view of the central role of VaR in practical risk assessment, it is natural to ask whether our method can be adapted to yield a viable elicitation procedure in this special case. 
This section pursues this goal, which requires greater care. 

\subsection{Eliciting the reference probability of VaR-capacities}

Throughout this section, $\Sigma$ is an algebra and $P$ a convex-ranged probability charge. 
The first tweak is that, in addition to the Value-at-Risk defined in Sect.~\ref{ex:VaR},  
we may need to consider the so-called right quantile, or right VaR, i.e., the functional
\[\overline{\var}^P_\gamma(X):=\inf\{x\in\R\colon P[X\le x]>\gamma\},\quad X\in B_s(\Sigma).\]
This change of perspective comes at no loss. 
If we focus on indicators $\ind_A$ of events $A\in \Sigma$, one observes that 
\begin{equation}\label{eq:duality}\overline{\var}_{1-\gamma}^P(\ind_A)=1-\var_\gamma^P(\ind_{A^c}),\quad A\in \Sigma,\end{equation}
i.e., the two classes of capacities are dual to each other. Knowing the values of one of them is sufficient to elicit the reference probability of both of them.

Moreover, it is easy to see that 
\begin{equation}\label{eq:iff}\rvar_\gamma^P(\ind_A)=\begin{cases}1&\quad\text{if }P[A]\ge 1-\gamma,\\
    0&\quad\text{else}.\end{cases}\end{equation}
All subsequently appearing subsets of $\Omega$ are events in the underlying algebra $\Sigma$.

Our approach assumes knowledge of the values a VaR-capacity assigns to all measurable subsets of $\Omega$, while treating {\em both} the reference probability $P$ and parameter  $\gamma\in(0,1)$ in $\var_\gamma^P$ as unknown. 
As a first step, this requires distinguishing between a {\em small $\gamma$ case} ($\gamma\le \frac 1 2$) and a {\em large $\gamma$ case} ($\gamma>\frac 1 2$). 
This distinction can be accomplished by the following test based on \eqref{eq:vgamma}.

\begin{lemma}\label{lem:test}
    Let $0<\gamma<1$. Then, there exists $A\in\Sigma$ with $\var_\gamma^P(\ind_A)=\var_\gamma^P(\ind_{A^c})=0$ if and only if $\gamma\le \tfrac 1 2$.
\end{lemma}

The second step of our approach recursively constructs a new set function on the basis of the original VaR-capacity. 
The construction procedure distinguishes two cases, and the test in Lemma~\ref{lem:test} determines which case to follow.  

\subsubsection{Small $\gamma$ case}

If $\gamma \le \tfrac 1 2$, we recursively define a family of set functions $(g_t)_{t\in \N_0}$ on $\Sigma$ by 
\begin{equation}\label{eq:def}\begin{cases}g_0(A) = \rvar_{1-\gamma}^P(\ind_A)=1-\var_\gamma^P(\ind_{A^c}),\\
g_{t}(A) =    \sup_{B\in \Sigma} \inf_{C\subseteq A^c}[g_{t-1} (A\cup B) + g_{t-1} (A\cup C) -g_{t-1}(B\cup C)]\wedge 1,\quad t\in \N.\end{cases}
\end{equation}

Note that the only information we need to perform the recursion is the values of $\var_\gamma^P$.

\begin{lemma}\label{lem:VaR1}
Suppose that $\gamma\le \tfrac 1 2$.
For $t\in \N_0$, let $g_t$ be defined by \eqref{eq:def}. 
Then the following statements hold: 
\begin{enumerate}
    \item[(i)] For all $A,A'\in \Sigma$, $P[A]\le P[A']$ implies $g_t(A)\le g_t(A')$;
    \item[(ii)] $g_t(\varnothing)=0$;
    \item[(iii)] it holds that $$g_t(A)=\begin{cases}1&\quad\text{if }P[A]\ge  2^{-t}\gamma,\\
0&\quad\text{else,}\end{cases}\qquad A\in\Sigma.$$
\end{enumerate}
\end{lemma}
\begin{proof}
We proceed by induction over $t$.
For $t=0$ we have $g_0=\rvar_{1-\gamma}^P$, and (iii) follows from \eqref{eq:iff}. 
Moreover, statements (i) and (ii) clearly hold. 

Suppose now that (i)--(iii) hold true for $t=0,\dots,n-1$.
We first verify statement (iii) for $g_n$, beginning with the fact that $g_n$ only attains values 0 and 1. 
To this effect, let $A,B,C\in\Sigma$ be arbitrary and observe 
$$[g_{n-1} (A\cup B) + g_{n-1} (A\cup C) -g_{n-1}(B\cup C)]\wedge 1\in\{-1,0,1\}.$$
Choosing $B=\Omega$ in the supremum 
part in \eqref{eq:def} and observing that $g_{n-1}(\Omega)=1$ by induction hypothesis, $g_n(A)\in\{0,1\}$ for all $A\in\Sigma$.

Now, we have to prove for arbitrary $A\in \Sigma$ that 
\begin{equation}\label{eq:prob}P[A]\ge 2^{-n}\gamma\quad\iff\quad g_n(A)=1.\end{equation}
To this effect, we distinguish three cases. First, if $P[A] \ge 2^{-n+1}\gamma$, choose $B=A$ to infer 
 $$\inf_{C\subseteq A^c}[g_{n-1} (A\cup B) + g_{n-1} (A\cup C) -g_{n-1}(B\cup C)]=g_{n-1} (A)=1.$$
 This shows $g_n(A)=1$.

Second, if $2^{-n}\gamma\le  P[A] <2^{-n+1}\gamma$, we select $B\subseteq A^c$
such that $P[A\cup B]=2^{-n+1}\gamma$. 
In particular, we have for each $C\subseteq A^c$ that 
$$P[A\cup C]=P[A]+P[C]\ge P[B]+P[C]\ge P[B\cup C].$$
Using (i) and (iii) for $g_{n-1}$,
$$ \inf_{C\subseteq A^c}[g_{n-1} (A\cup B) + g_{n-1} (A\cup C) -g_{n-1}(B\cup C)] 
\ge g_{n-1} (A\cup B) =1,
$$
which suffices to prove that $g_n(A)=1$. 

Third, suppose $P[A]< 2^{-n}\gamma$.
Let $B\in\Sigma$ be arbitrary. 
If $P[B]\le P[A]$, we obtain from 
$P[A\cup B]<2^{-n+1}\gamma$ that
\begin{align*}
 &\inf_{C\subseteq A^c}[g_{n-1} (A\cup B) + g_{n-1} (A\cup C) -g_{n-1}(B\cup C)]\\
 \le~&g_{n-1}(A\cup B) + g_{n-1} (A)-g_{n-1}(B)=0.
\end{align*}
Else, if $P[B] \ge 2^{-n+1}\gamma$, then
\begin{align*}
 &\inf_{C\subseteq A^c}[g_{n-1} (A\cup B) + g_{n-1} (A\cup C) -g_{n-1}(B\cup C)]\\
 \le~&g_{n-1} (A\cup B) + g_{n-1}(A) -g_{n-1}(B)\\
=~&1+0-1=0.
\end{align*}
Last, if 
$2^{-n+1}\gamma>P[B]>P[A]$, 
take $D\subseteq (A\cup B)^c$ such that $P[D\cup B]=2^{-n+1}\gamma$.
Note that such an event $D$ exists since $\gamma\le \tfrac 1 2$.
As $P[A\cup D]<2^{-n+1}\gamma$ holds by construction, 
\begin{align*}&\inf_{C\subseteq A^c}[g_{n-1} (A\cup B) + g_{n-1} (A\cup C) -g_{n-1}(B\cup C)]\wedge 1\\
\le~&g_{n-1}(A\cup B) + g_{n-1}(A\cup D) -g_{n-1}(B\cup D)\\
\le~&1 +0  -1 
 = 0.
 \end{align*}

In summary, we conclude $g_n(A)\le 0$, and thus $g_n(A)=0$ since it cannot be negative. This completes the proof of equivalence \eqref{eq:prob}. Statements (i) and (ii) of the lemma follow immediately from statement (iii). 
\end{proof}

Based on Lemma \ref{lem:VaR1}, we introduce a new set function $v$ on $\Sigma$ by 
\begin{align}\label{eq:v1}
v(A)&=\begin{cases}\sup\{2^{-t} \colon t\in\N_0\text{ and }g_t (A)=1\} &\quad\text{if }g_t(A)=1\text{ for some }t\in\N_0\\[-0.3ex]
0&\quad\text{else}\end{cases}\nonumber\\
&= \begin{cases}\sup\{2^{-t} \colon t\in\N_0\text{ and }P[A]\ge 2^{-t}\gamma\}&\quad\text{if }P[A]>0\\[-0.3ex]
0&\quad\text{else}.\end{cases}
\end{align} 
The set function
$v$ can be interpreted as an approximation of the size of the $P$-probability of an event $A$ that becomes more and more accurate the smaller that probability is.

Finally, in the third step of the modified elicitation procedure, we demonstrate that $v$ is a game for which the reference probability can be elicited using the methods outlined in Sect.~\ref{sec:distortion}.

\begin{proposition}\label{prop:VaR1}
Suppose that $\gamma\le \frac 1 2$ and that game $v$ is given by \eqref{eq:v1}. 
Then the loose core $\LC_v$ satisfies
$$\inf\LC_v=\tfrac 1 \gamma P,$$
i.e., for $\mu_\star:=\inf\LC_v$, we have  $P=\mu_\star(\Omega)^{-1}\mu_\star$and $\gamma=\frac{1}{\mu_\star(\Omega)}.$
\end{proposition}
\begin{proof}
By construction, $v\le\frac 1 \gamma P$ and the associated loose core $\LC_v$ is non\-empty. 
Next, pick arbitrary $\mu\in\LC_v$ and $A\in\Sigma$.
If $P[A]=0$, we have 
$$\mu(A)\ge v(A)=0=\tfrac 1 \gamma P[A].$$
If $P[A]>0$, let $n\in\N$ large enough such that we can partition $A$ into pairwise disjoint subevents $A_1,\dots, A_{m_n+1}$ with the property 
$P[A_i]=2^{-n}\gamma$, $1\le i\le m_n$, and $P[A_{m_n+1}]<2^{-n}\gamma$.
The set $A_{m_n+1}$ may be empty. 
Then, 
$$\mu(A)=\sum_{i=1}^{m_n+1}\mu(A_i)\ge \sum_{i=1}^{m_n+1}v(A_i)\ge m_n 2^{-n}=\frac{m_n 2^{-n}\gamma}{\gamma}.$$
As $n\to\infty$, we obtain 
$$\mu(A)\ge \tfrac{P[A]}{\gamma},$$
which suffices to show that $\inf\LC_v=\frac 1 \gamma P$. 
\end{proof}

\subsubsection{Large $\gamma$ case}

In case $\gamma\in(\frac 1 2,1)$, the second step of our modified elicitation procedure  requires a different recursion based on the family $(h_t)_{t\in \N_0}$ of set functions on $\Sigma$ defined by
\begin{equation}\label{eq:defh}
\begin{cases}h_0(A):=\VaR^P_\gamma(\ind_{A}),\\
h_t(A)=\sup_{B\in \Sigma} \inf_{C\subseteq A^c}[h_{t-1} (A\cup B) + h_{t-1} (A\cup C) -h_{t-1}(B\cup C)]\wedge 1,\quad t\in\N.\end{cases}
\end{equation}
Again, the recursion---and thus the computation of $w$ below---requires only the values of $\var_\gamma^P$ as initial input. 
The following lemma is the analogue of Lemma~\ref{lem:VaR1}.
Given that its proof is structurally similar to the preceding one, we will omit the detailed exposition.

\begin{lemma}
Suppose that $\gamma>\tfrac 1 2$.
For $t\in \N_0$, let $h_t$ be defined by \eqref{eq:defh}. 
Then the following statements hold: 
\begin{enumerate}
    \item[(i)] For all $A,A'\in \Sigma$, $P[A]\le P[A']$ implies $h_t(A)\le h_t(A')$; 
    \item[(ii)] $h_t(\varnothing)=0$;
    \item[(iii)] it holds that $$h_t(A)=\begin{cases}1&\quad\text{if }P[A]>2^{-t}(1-\gamma),\\[-0.5ex]
0&\quad\text{else,}\end{cases}\qquad A\in\Sigma.$$
\end{enumerate}
\end{lemma}

In analogy with \eqref{eq:v1}, we set 
\begin{align}\label{eq:v3}
    w(A)&=\begin{cases}\sup\{2^{-t}\colon t\in\N_0\text{ and }h_t (A)=1\} &\quad\text{if }h_t(A)=1\text{ for some }t\in\N_0\\[-0.3ex]
0&\quad\text{else}\end{cases}\nonumber\\
&=\begin{cases}\sup\{2^{-t}\colon t\in\N_0\text{ and }P[A]>2^{-t}(1-\gamma)\}&\quad\text{if }P[A]>0\\[-0.3ex]
0&\quad\text{else}.\end{cases}
\end{align}
Parallel to Proposition~\ref{prop:VaR1}, we can apply the methodology from Sect.~\ref{sec:distortion} to $w$. 

\begin{proposition}\label{prop:VaR2}
Suppose that $\gamma>\frac 1 2$ and define game $w$ by \eqref{eq:v3}. 
Then the loose core $\LC_w$ satisfies 
\[\inf\LC_w=\tfrac 1{1-\gamma}P,\]
i.e., for $\nu_\star:=\inf\LC_w$, we have $P=\nu^\star(\Omega)^{-1}\nu^\star$and $\gamma=1-\tfrac1{\nu^\star(\Omega)}.$
\end{proposition}

\begin{proof}
By construction, $w\le\frac 1{1-\gamma}P$, i.e., the associated loose core $\LC_w$ is non\-empty. 
Next, pick arbitrary $\mu\in\LC_w$ and $A\in\Sigma$.
If $P[A]=0$, we have 
$$\mu(A)\ge w(A)=0=\tfrac 1{1-\gamma}P[A].$$
If $P[A]>0$, let $n\in\N$ large enough such that we can partition $A$ into pairwise disjoint subevents $A_1,\dots, A_{m_n}$ with the properties 
$2^{-n}(1-\gamma)<P[A_i]\le 2^{-n+1}(1-\gamma)$, $1\le i\le m_n$, and $\lim_{n\to\infty}m_n2^{-n}(1-\gamma)=P[A]$. 
Then, 
$$\mu(A)=\sum_{i=1}^{m_n}\mu(A_i)\ge \sum_{i=1}^{m_n}w(A_i)=m_n 2^{-n}=\frac{m_n 2^{-n}(1-\gamma)}{1-\gamma}.$$
As $n\to\infty$, we obtain $\mu(A)\ge \tfrac{P[A]}{1-\gamma}$,
which suffices to show that $\inf\LC_w=\frac 1{1-\gamma}P$. 
\end{proof}

\subsection{Comparison to axiomatisations of quantiles}

With Propositions~\ref{prop:VaR1} and \ref{prop:VaR2}, we have solved the problem of identifying the reference probability for (non-degenerate) VaR capacities.
In this subsection, we briefly compare our solution to existing work on axiomatising quantile preferences, focusing in particular on the paper by Rostek~\cite{Rostek}.
In contrast, the approach of Chambers~\cite{Chambers},  recently strengthened by Fadina et al.~\cite{OneAxiom}, treats VaR as a functional on distribution functions. The latter perspective differs fundamentally from ours, which is grounded in random variables and events, and avoids the challenge of eliciting a reference probability altogether.

The setting in Rostek~\cite{Rostek}, however, shares some features with ours.
This setting concerns a ``Savagean model of purely
subjective uncertainty'', where acts play the role of random variables. 
One of the key goals of the paper (see \cite[Theorem 1]{Rostek}) is to characterise axiomatically---up to complications created by degenerate cases---whether a given preference relation has a numerical representation by a quantile function; that is, by $\var_\gamma^P$ for a suitable parameter $\gamma$ and a convex-ranged probability charge $P$. 

In order to prove Theorem 1, Rostek~\cite{Rostek} needs to answer three questions: Does the numerical representation belong to the VaR-class?
 What is the (convex-ranged) reference probability $P$?
And what is the parameter $\gamma$?
The elicitation we perform in the present section only addresses the last two questions and presumes that the first one is answered affirmatively. 
Thus, our contribution can be seen as a step in the broader framework of the proof of \cite[Theorem 1]{Rostek}. 
The latter, however, is very intricate and uses the machinery of Fishburn's \cite[Chapter 14]{Fishburn} derivation of subjective expected utility. 
Our approach is  comparatively direct, yet firmly grounded in the general principles explored in Sects.~\ref{sec:first} and \ref{sec:distortion}.
How one would obtain Propositions~\ref{prop:VaR1} and \ref{prop:VaR2} from \cite{Rostek} more easily---if possible at all---is unclear to us.

\section{Conclusion}

Our paper studies the problem of finding the {\em a priori} unknown reference measure of a functional that we suspect to be law invariant. 
In a nutshell, the results we prove target lower (upper) supporting sets of that functional, signed measures---or signed charges---whose integrals are pointwise bounded above (below) by the functional in question. 
It is then shown that the supremum (infimum) of this set in the vector lattice of signed charges---if it exists---is a multiple of the reference measure.
This multiple may be zero. 
In cases where it is, the results make it possible to pin down the only possible candidate for the reference measure and to potentially disprove law invariance of the functional altogether. 

In the important case of the Value-at-Risk (quantile functionals), the previous approach does not work directly. 
However, it delivers a neat way to elicit the reference measure when combined with some additional subtle steps.

The results obtained in the paper are of a theoretical nature, and we acknowledge that their implementation  in practical applications is largely unaddressed. We leave this important aspect to future research. 
A natural direction would be to investigate how finding the candidate probability measures can be operationalised or realised algorithmically, especially within the context of financial data and regulatory schemes.

\appendix

\section{Preliminaries on charges and set functions}\label{sec:ancillary}

\subsection{Definitions}\label{app:definitions}

Let $\Sigma$ be an algebra on $\Omega$. 
The real vector space $\ba$ collects all signed charges $\mu\colon\Sigma\to\R$ for which $\sup_{A\in\Sigma}|\mu(A)|<\infty$.
Equipped with the setwise order $\mu\le \nu$, which holds if $\mu(A)\le \nu(A)$ is satisfied by all $A\in\Sigma$, it is a {\em Dedekind complete}
vector lattice, i.e., every upper bounded subset has a supremum; see \cite[Theorem 10.53]{Ali}. 
We also write $|\mu|:=\mu\vee(-\mu)$.
Next, for $\mu\in\ba$ and $\nu\in\ba_+:=\{\mu\in\ba\colon \mu\ge 0\}$, we write 
$\mu\ll\nu$ if $\nu(A)=0$ implies $\mu(A)=0$.
If $\Sigma$ is a $\sigma$-algebra, the subspace $\ca$ of countably additive signed measures therein is a so-called band. That is, every supremum of upper bounded subsets of $\ca$ lies in $\ca$ itself, and $\ca$ is an \emph{ideal} in the sense that if $\mu \in \ca$ and $|\nu| \leq |\mu|$, then $\nu \in \ca$ as well; see \cite[Chapter 8.9]{Ali}.

A game $v$ on $\Sigma$ is {\em submodular} if, for all $A,B\in\Sigma$, 
\begin{equation}\label{eq:submodular}v(A)+v(B)\ge v(A\cap B)+v(A\cup B).\end{equation}
If $v$ has bounded variation (see Marinacci and Motrucchio \cite{MariMont}), then submodularity is equivalent to subadditivity of the associated Choquet integral. 
In particular, every {\em capacity}, i.e., every game that is nondecreasing with respect to set inclusion, has bounded variation. 

\subsection{Ancillary results}

Throughout the rest of this section, we work in the setting of Assumption~\ref{ass:structure}(B) with a convex-ranged probability charge $P$ over an algebra $\Sigma$. Moreover, $\mu$ denotes an element of $\ba$.

\begin{lemma}\label{lem:sup}
Suppose that $\mu\in\ba$ satisfies $\mu\ll P$ and define $\rho_\mu\colon B_s(\Sigma)\to\R$ by 
    \begin{equation}\label{rho_mu}\rho_\mu(X)=\sup_{Y\sim_P X}\E_\mu[Y].\end{equation}
    Then $\rho_\mu$ is a subadditive, positively homogeneous, $P$-invariant, and continuous functional. 
    If $\Sigma$ is a $\sigma$-algebra and $P$ is replaced by an atomless probability measure $\P$, extending the defining equation \eqref{rho_mu} to $B(\Sigma)$ also provides a subadditive, positively homogeneous, $\P$-invariant, and continuous functional. 
\end{lemma}
\begin{proof}
    From $\mu\ll P$, we infer $\rho_\mu(0)=0$. 
    Now fix $X,Y\in B_s(\Sigma)$ and $X'\sim_P X$. 
    The convex range of $P$ permits selecting $Y'\sim_P Y$ such that $\|X-Y\|_\infty\ge \|X'-Y'\|_\infty$. For this pair $(X',Y')$, we obtain 
    $$E_\mu[X']\le \E_\mu[Y']+\E_{|\mu|}[|X'-Y'|]\le \rho_\mu(Y)+|\mu|(\Omega)\|X-Y\|_\infty.$$
    Together with $\rho_\mu(0)=0$, this is sufficient to show that $\rho_\mu$ only takes finite values and is continuous. 
    Positive homogeneity of $\rho_\mu$ is clear from its definition. For subadditivity, one should note that for every $Z\sim_P X+Y$ the convex range of $P$ allows to select $X^Z\sim_P X$ and $Y^Z\sim_P Y$ such that $X^Z+Y^Z=Z$.
\end{proof}

Now we introduce the set functions $s_\mu$ and $\iota_\mu$ on $\Sigma$ by 
\[s_\mu(A)=\sup\{\mu(B)\colon B\in\Sigma,P[B]=P[A]\}\]
and
\[\iota_\mu(A)=\inf\{\mu(B)\colon B\in\Sigma,P[B]=P[A]\}.\]
These are conjugate to each other via the relation
$$s_\mu(A)=\mu(\Omega)-\iota_\mu(A^c),\quad A\in\Sigma,$$ 
and in the special case $\mu\ll P$, we have 
$$s_\mu(A)=\rho_\mu(\ind_A)\quad\text{and}\quad \iota_\mu(A)=-\rho_\mu(-\ind_{A}),\quad A\in\Sigma.$$

\begin{lemma}[Theorem 8 in \cite{Basket1}]\label{lem:FH}
Suppose that $\mu\in\ba$ is linearly independent of $P$. 
Then, for all $A\in\Sigma$ with $P[A]\in(0,1)$,
\[\iota_\mu(A)<\mu(\Omega)P[A]<s_\mu(A).\]
\end{lemma}

The final ancillary result in this appendix provides a sufficient (and necessary) condition under which $s_\mu$ is a submodular game of bounded variation.
Another term that appears in its statement is {\em exactness}. 
A game $v$ is exact if, for all $A\in\Sigma$, $v(A)=\sup_{\mu\in\acore_v}\mu(A)$. 

\begin{lemma}\label{lem:concave}
Let $\mu\in\ba$ such that its positive part $\mu^+=\mu\vee 0$ satisfies $\mu^+\ll P$. Then $s_\mu$ is a submodular exact game of bounded variation. Moreover, for all $A\in\Sigma$ with $P[A]\in(0,1)$, 
$$s_\mu(A)\ge \mu(\Omega)P[A].$$
\end{lemma}

\begin{proof}
    By \cite[Lemma 7]{Basket1}, the set function $s_\mu$ is submodular in the sense that it satisfies \eqref{eq:submodular}. 
    However, $s_\mu$ is also a game, i.e., $s_\mu(\varnothing)=0$: 
    \begin{align*}0&=\mu(\varnothing)\le s_\mu(\varnothing)=\sup\{\mu(N)\colon  N\in\Sigma,\,P[N]=0\}\\
        &\le \sup\{\mu^+(N)\colon N\in\Sigma,\,P[N]=0\}=0.
    \end{align*}
    
Next, we observe that $s_\mu$ is a bounded game. 
Indeed,  
$$\sup_{A\in\Sigma}|s_\mu(A)|\leq \sup_{A\in\Sigma}|\mu(A)|<\infty.$$
By \cite[Theorem 4.7]{MariMont}, $s_\mu$ is of bounded variation and exact. 

For the verification of the last assertion, note that the Choquet integral $\ph$ on $B_s(\Sigma)$ with respect to $s_\mu$ is sublinear and $P$-invariant.
Using \cite[Lemma 8]{Basket1}, we have for all $A\in\Sigma$ with $P[A]\in(0,1)$ that 
$$s_\mu(A)=\ph(\ind_A)\ge \ph(1)P[A]\ge \mu(\Omega)P[A].$$
\end{proof}

\subsection*{Competing Interests}
The authors declare no competing interests.

\subsection*{Acknowledgements}
The authors thank two anonymous reviewers and Peter Wakker for useful comments on an earlier version of the paper.


\begin{thebibliography}{100}
\singlespacing
\itemsep0em 
\scriptsize

\bibitem{Acerbi}Acerbi, C.\ (2002), Spectral measures of risk: A coherent representation of subjective risk aversion. \textit{Journal of Banking and Finance} 26(7):1505--1518.

\bibitem{Ali}Aliprantis, C.\ D., and K.\ C.\ Border (2006), \textit{Infinite Dimensional Analysis: A Hitchhiker's Guide}. 3rd edition, Springer.

\bibitem{Sandwich}Amarante, M.\ (2019), Sandwich Theorems for set functions: An application of the Lehrer-Teper integral. \textit{Fuzzy Sets and Systems} 355:59--66.

\bibitem{Basket1}Amarante, M., F.-B.\ Liebrich, and C.\ Munari (2024), Uniqueness of convex-ranged probabilities and applications to risk measures and games. {\em Mathematics of Operations Research} 50(1):743-763. 

\bibitem{AA}Anscombe, F.\ J., and R.\ J.\ Aumann (1963), A definition of subjective probability. \textit{Annals of Mathematical Statistics} 34(1):199--205.

\bibitem{Bellini}
Bellini, F., P.\ Koch-Medina, C.\ Munari,  and G.\ Svindland (2021),
Law-invariant functionals on general spaces of random variables.
{\em SIAM Journal on Financial Mathematics} 12(1):318--341.

\bibitem{BR}Bhaskara Rao, K.\ P.\ S., and M.\ Bhaskara Rao (1983), \textit{Theory of Charges: A Study of Finitely Additive Measures}. Academic Press.

\bibitem{Chambers}Chambers, C.\ P.\ (2009), An axiomatization of quantiles on the domain of distribution functions.
{\em Mathematical Finance} 19(2):335--342.

\bibitem{Cont}Cont, R., R.\ Deguest, and G.\ Scandolo (2010), Robustness and sensitivity analysis of risk measurement procedures. \textit{Quantitative Finance} 10(6):593--606.

\bibitem{Finetti}de Finetti, B.\ (1931), Sul significato soggettivo della probabilità. {\em Fundamenta Mathematicae} 17:298--329. 

\bibitem{EMWW21}
 {Embrechts, P., T. Mao, Q. Wang, and  R. Wang} (2021), Bayes risk, elicitability, and the Expected Shortfall. \emph{Mathematical Finance} {31}:1190--1217. 

\bibitem{OneAxiom}Fadina, T., P.\ Liu, and R.\ Wang (2023), One axiom to rule them all: A minimalist axiomatization of quantiles.
{\em SIAM Journal on Financial Mathematics} 14(2):644--662.


\bibitem{Fadina}Fadina, T., Y.\ Liu, and R.\ Wang (2024), A framework for measures of risk under uncertainty. {\em Finance and Stochastics} 28:363--390.

\bibitem{FilSvi}Filipovi\'c, D., and G.\ Svindland (2012), The canonical model space of law invariant risk measures is $L^1$. {\em Mathematical Finance} 22(3):585--589.

\bibitem{Fishburn}Fishburn, P.\ C.\ (1970), {\em Utility Theory for Decision Making}. Wiley, New York.


\bibitem{FZ16}
Fissler, T., and  J.\ F.\ Ziegel (2016), Higher order elicitability and Osband's principle. \emph{Annals of Statistics} {44}(4):1680--1707.

\bibitem{FoeSch}F\"ollmer, H., and A.\ Schied (2016), \textit{Stochastic Finance: An Introduction in Discrete Time}. 4th edition, De Gruyter. 

\bibitem{Frittelli}Frittelli, M., and E.\ Rosazza Gianin (2005), Law invariant convex risk measures. {\em Advances in Mathematical Economics} 7:33--46.

\bibitem{Gilboa}Gilboa, I.\ (1987), Expected utility with purely subjective non-additive probabilities. \textit{Journal of Mathematical Economics} 16(1):65--88.

\bibitem{Review}He, X.\ D., S.\ Kou, and X.\ Peng (2022), Risk measures: Robustness, elicitability, and backtesting. \textit{Annual Review of Statistics and Its Applications} 9:141--166.

\bibitem{JST}Jouini, E., W.\ Schachermayer, and N.\ Touzi (2006), Law invariant risk measures have the Fatou property. \textit{Advances in Mathematical Economics} 9:49--71.

\bibitem{KW88}
Kadane, J.\ B., and R.\ L.\  Winkler (1988), Separating probability elicitation from utilities. \emph{Journal of the American Statistical Association} 83(402):357--363.


\bibitem{Kindler}Kindler, J.\ (1988), Sandwich theorems for set functions. {\em Journal of Mathematical Analysis and Applications} 133:529--542.

\bibitem{KSZ}Kr\"atschmer, V., A.\ Schied, and H.\ Zähle (2014), Comparative and qualitative robustness for law-invariant risk measures. \textit{Finance and Stochastics} 18:271--295.


\bibitem{KP16}
{Kou, S., and  X. Peng} (2016), On the measurement of economic tail risk. \emph{Operations Research} {64}(5):1056--1072.

\bibitem{K01}
{Kusuoka, S.} (2001), On law invariant coherent risk measures. \emph{Advances in Mathematical Economics} {3}:83--95.






\bibitem{Lehrer}Lehrer, E., and R. Teper (2008),
The concave integral over large spaces. {\em Fuzzy Sets and Systems} 159:2130--2144.

\bibitem{Liebrich}Liebrich, F.-B.\ (2024), Are reference measures of law-invariant functionals unique? {\em Insurance: Mathematics and Economics} 118:129--141. 

\bibitem{LiebrichMunari}Liebrich, F.-B., and C.\ Munari (2025), Revisiting the automatic Fatou property of law-invariant functionals. 
{\em SIAM Journal on Financial Mathematics} 16(1):SC1--SC11.

\bibitem{MJ}Machina, M.\ J., and D.\ Schmeidler (1992), A more robust definition of subjective probability. \textit{Econometrica} 60(4):745--780.

\bibitem{MariMont}Marinacci, M., and L.\ Montrucchio (2004), Introduction to the mathematics of ambiguity. In: \textit{Uncertainty in
Economic Theory: a Collection of Essays in Honor of David Schmeidler’s 65th Birthday}.

\bibitem{MFE15}
{McNeil, A.\ J., R.\ Frey, and P.\ Embrechts} (2015), \emph{Quantitative
Risk Management: Concepts, Techniques and Tools}. Revised edition.  Princeton, NJ:
Princeton University Press.

\bibitem{MeyNie}Meyer-Nieberg, P.\ (1991), \textit{Banach Lattices}. Springer.

\bibitem{Quiggin}Quiggin, J.\ (1982), A theory of anticipated utility. \textit{Journal of Economic Behavior \& Organization} 3(4):323--343.

\bibitem{Ramsey} Ramsey, F.\ P.\ (1926), Truth and probability. In: Ramsey, F.\ P.\ (1931), \textit{The Foundations of Mathematics and other Logical Essays}, p.\ 156--198.



\bibitem{RUZ06}
Rockafellar, R. T., S. Uryasev, and M. Zabarankin (2006), Generalized deviation in risk analysis. \emph{Finance and Stochastics} 10:51--74.



\bibitem{Rostek}Rostek, M.\ (2010), Quantile maximization in decision theory. {\em Review of Economic Studies} 77:339--371.

\bibitem{Savage} Savage, L.\ J.\ (1954), {\em The Foundations of Statistics}. Wiley.

\bibitem{Schmeidler}
{Schmeidler, D.} (1989), Subjective probability and expected utility without additivity.
\emph{Econometrica} {57}(3):571--587.


\bibitem{Shapiro}Shapiro, A.\ (2013), Consistency of sample estimates of risk averse stochastic programs. \textit{Journal of Applied Probability} 50(2):533--541.

\bibitem{Svindland}
Svindland, G.\ (2010), Continuity properties of law-invariant (quasi-)convex risk functions on $L^\infty$. \textit{Mathematics and Financial Economics} 3:39--43.

\bibitem{Riskmetrics}Wang, Q., R.\ Wang, and Y. Wei (2020), Distortion riskmetrics on general spaces. 
\textit{ASTIN Bulletin} 50(3):827--851.

\bibitem{WWW}Wang, R., Y.\ Wei, and G.\ Willmot (2020), Characterization, robustness and aggregation of signed Choquet integrals. \textit{Mathematics of Operations Research} 45(3):993--1015.

\bibitem{Insurance}Wang, S., V.\ R.\ Young, and H.\ H.\ Panjer (1997), Axiomatic characterization of insurance prices. 
{\em Insurance:
Mathematics and Economics} 21(2):173--183.



\bibitem{W06}
{Weber, S.} (2006), Distribution-invariant risk measures, information, and dynamic consistency. \emph{Mathematical Finance} 16:419--441.




\bibitem{Yaari}Yaari, M.\ E.\ (1987), The dual theory of choice under risk. \textit{Econometrica} 55(1):95--115.


\bibitem{Z16}
Ziegel, J. (2016), Coherence and elicitability. {\em Mathematical Finance} 26:901--918.




\end{thebibliography}
\end{document}